\documentclass[journal]{IEEEtran}
\usepackage{cite}
\usepackage{amsmath,amssymb,amsfonts}
\usepackage{algorithmic}
\usepackage{graphicx}
\usepackage{textcomp}
\usepackage{xcolor}

\usepackage{color}

\usepackage{soul}

\soulregister\cite7
\soulregister\ref7
\soulregister\eqref7
\soulregister\label7
\soulregister\prop7

\usepackage{amsthm,amssymb}
\usepackage[english]{babel}
\usepackage{mathalfa}
\usepackage[para]{footmisc}

\usepackage{caption}
\usepackage{epstopdf}
\usepackage{tabularx}
\usepackage[section]{placeins}
\usepackage{cuted}

\usepackage{setspace}
\usepackage{stfloats}
\usepackage{multirow}
\usepackage{float}
\usepackage{algorithm}

\usepackage{multirow}                

\newtheorem{theorem}{Theorem}

\newtheorem{remark}{Remark}
\newtheorem*{Proof}{Proof}
\usepackage{xcolor} 
\usepackage{enumerate}

\newtheorem{prop}{Proposition}
\usepackage{color}

\def\BibTeX{{\rm B\kern-.05em{\sc i\kern-.025em b}\kern-.08em
		T\kern-.1667em\lower.7ex\hbox{E}\kern-.125emX}}

\ifCLASSINFOpdf
\else
\fi

\begin{document}
%
\title{UAV Swarm-Enabled Aerial CoMP: A Physical Layer Security Perspective}
%
%
%

\author{Xuanxuan Wang,
        Wei Feng,~\IEEEmembership{Senior Member,~IEEE,}
        Yunfei Chen,~\IEEEmembership{Senior Member,~IEEE,}\\
        Ning Ge,~\IEEEmembership{Member,~IEEE,}
\thanks{This work was supported in part by the Beijing Natural Science Foundation
	(Grant No. L172041), in part by the National Natural Science Foundation of
	China (Grant No. 61771286, 61701457, 91638205), and in part by the Beijing
	Innovation Center for Future Chip.}
\thanks{X. Wang and N. Ge are with Beijing National Research Center for Information Science and Technology, Tsinghua University, Beijing, China, and
	Department of Electronic Engineering, Tsinghua University, Beijing, China
	(e-mail: wangxuanxuan@tsinghua.edu.cn; gening@tsinghua.edu.cn).}
\thanks{W. Feng is with Peng Cheng Laboratory, Shenzhen, China, Beijing National
	Research Center for Information Science and Technology, Tsinghua University, Beijing, China, and Department of Electronic Engineering, Tsinghua University, Beijing, China (e-mail: fengwei@tsinghua.edu.cn;). (Corresponding
	author: Wei Feng)}
\thanks{Y. Chen is with School of Engineering, University of Warwick, Coventry
	CV4 7AL, U.K. (e-mail: Yunfei.Chen@warwick.ac.uk).}}

\maketitle

\begin{abstract}
	Unlike aerial base station enabled by a single unmanned aerial vehicle (UAV), aerial coordinated multiple points (CoMP) can be enabled by a UAV swarm. In this case, the management of multiple UAVs is important. This paper considers the power allocation strategy for a UAV swarm-enabled aerial network to enhance the physical layer security of the downlink transmission, where an eavesdropper moves following the trajectory of the swarm for better eavesdropping. Unlike existing works, we use only the large-scale channel state information (CSI) and maximize the secrecy throughput in a whole-trajectory-oriented manner. The overall transmission energy constraint on each UAV and the total transmission duration for all the legitimate users are considered. The non-convexity of the formulated problem is solved by using max-min optimization with iteration. Both the transmission power of desired signals and artificial noise (AN) are derived iteratively. Simulation results are presented to validate the effectiveness of our proposed power allocation algorithm and to show the advantage of aerial CoMP by using only the large-scale CSI.
\end{abstract}

\begin{IEEEkeywords}
Artificial noise, large-scale fading, physical layer security, secrecy throughput, UAV swarm
\end{IEEEkeywords}

%
\IEEEpeerreviewmaketitle

\section{Introduction}
In recent years, unmanned aerial vehicles (UAVs) have attracted great interest in wireless communications \cite{Zeng-2016Magazine,Koulali-2016Magaziney,Zeng-2016TC,Xuan-JCIN2018}. Due to their mobility and elevated position, they can provide agile communications \cite{Wang-2016VTM}. With their high maneuverability, UAV can augment the network capacity and coverage, especially in the extreme environments without infrastructure, such as disaster rescue, traffic monitoring and so on \cite{Zeng-2016Magazine,Yin-2018IoTJournal}. More specifically, UAVs are usually cost-effective \cite{Zeng-TWC2017,Chen-ICL2018,Xiang-arXiv2018,Qi-ChinaCom2018}. They can be exploited to assist on-demand missions, such as high-speed data transmission in the fifth generation (5G) wireless networks. In addition, with the huge demand in emergency applications, i.e., public safety, delivery and surveillance, deploying a flock of UAVs, or swarm, is becoming more attractive, which plays a vital role in meeting performance requirements for communications between multiple UAVs and 5G \cite{Yuan-2018CM,Li-IoT2018,3GPP-22261,3GPP-22862,3GPP-36777}.  

One of the serious concerns in UAV swarm-enabled aerial networks is how to guarantee the privacy and secrecy of the system. Due to the broadcast nature and inherent randomness of wireless channels, UAV swarm-enabled communication networks are particularly vulnerable to various security threats, such as information eavesdropping, information leakage, data modification and so on. In addition, to facilitate the secure transmission, the UAV swarm often places itself near the legitimate users, which is beneficial to eavesdropping, especially when the eavesdropper moves close to the legitimate users. 

\subsection{Related Work}
To achieve perfect security, the conventional encryption schemes are typically implemented at the upper layer using cryptographic methods. However, this is often achieved at the cost of high computational complexity \cite{Strohmeier-2015CST}.

Unlike the traditional cryptographic methods, physical layer security (PLS), using the information-theoretic and signal processing approaches, has been widely investigated in the UAV-enabled wireless networks \cite{Wang-2017WCL,Lee-2018TVT,Cui-TVT2018,Zhou-2018 TVT}. They enhance the coverage and security of the wireless systems by exploiting physical characteristics of the wireless channel. Specifically, by adaptively adjusting the UAVs' location, they could overcome the propagation constraints in the cellular systems, and provide new possibilities or opportunities for security enhancement. The authors in \cite{Wang-2017WCL} utilized UAV as a mobile relay, and maximized the secrecy rate of the system with transmit optimization in a four-node. In \cite{Lee-2018TVT}, the authors investigated UAV-enabled secure communication systems where a mobile UAV sent confidential messages to multiple ground users. By considering the imperfect information on the locations of the eavesdroppers, the authors in \cite{Cui-TVT2018} investigated a UAV-ground communication system with multiple potential eavesdroppers on the ground.  The authors in \cite{Zhou-2018 TVT} considered UAV-assisted secure communications between a legitimate transmitter-receiver pair for unknown eavesdropper location by taking UAV as an air-to-ground friendly jammer.

These studies \cite{Wang-2017WCL,Zhou-2018 TVT,Lee-2018TVT,Cui-TVT2018} have provided insightful results for improving the secrecy performance of the UAV-aided wireless communications. However, they assume an ideal free-space path-loss model \cite{Wang-2017WCL,Lee-2018TVT,Cui-TVT2018} between the UAV and the legitimate receivers/eavesdroppers or the instantaneous channel state information (CSI) \cite{Zhou-2018 TVT} of the eavesdroppers at the transmitter, which may not be practical. 

In practice, it is generally difficult to acquire the instantaneous CSI of the eavesdroppers, especially when they are passive. To deal with that, an effective approach, named as artificial noise(AN), has been proposed to mask the desired signals for enhancing the secrecy performance \cite{Zhou-2010TVT,Zhang-2013TVT,Lin-2013JASC,Li-2013TSP,Zhu-TWC2016}, where AN is designed based on the instantaneous CSI of the legitimate receiver and transmitted in the null-space of the legitimate channel. Although this scheme is helpful for the security, it requires perfect instantaneous CSI between the source and the legitimate receiver at the transmitter, which is nearly unworkable. The idea is then generalized to the UAV-enabled wireless systems, where a UAV is applied as a mobile jammer to transmit AN \cite{Zhou-2019TVT} or a legitimate receiver \cite{Liu-2017WCSP}. However, these works haven't shown useful guidelines to improve physical layer security of UAV swarm-enabled aerial networks. 

\subsection{Main Contributions}
Despite of the above fruitful results, some challenges still remain in the UAV swarm-enabled aerial networks.

For the UAV swarm-enabled aerial networks, an open challenge is how to acquire CSI. To practically depict the typical propagation environments, the composite channel model, consisting of both small-scale and large-scale fading, needs to be used, which is in stark contrast to the existing literatures \cite{Wang-2017WCL,Lee-2018TVT,Cui-TVT2018}. Under the composite channel, one key role for the power allocation strategy is the prior knowledge. Since it is impossible to perfectly acquire the random small-scale fading prior to the whole trajectory of the UAV swarm, it is almost infeasible to assume perfect CSI. In this paper, we devote to guarantee the secrecy performance of the system in a whole-trajectory-oriented manner by utilizing only the large-scale CSI of the legitimate receivers/eavesdroppers, which can be achieved at much lower cost. 

In wireless communication systems, path loss could significantly reduce the signal reception quality at the legitimate users, especially in the UAV swarm-enabled aerial networks. In the existing literatures, one effective scheme to overcome the limitation is by means of multiple antenna systems, i.e., multiple-input single-output (MISO) \cite{Khisti-2010TIT-I,Shi-2014TWC}, multiple-input multiple-output (MIMO) \cite{Khisti-2010TIT-II,Li-2013JASC}, or single-input multiple-output (SIMO) \cite{Parada-ISIT2005}. However, due to the limited size, it is hard for UAVs to be equipped with multiple antennas. To handle that, we consider an effective coordinated multiple points (CoMP) between UAVs in this paper, where multiple single-antenna UAVs are combined to form the UAV swarm and then act as a virtual multiple-antenna node. Unlike the conventional CoMP with fixed base stations (BSs), the UAV swarm is able to cooperatively operate as an aerial CoMP by utilizing the mobility of the UAVs. Note that, in contrast to the existing works achieving CoMP based on perfect CSI \cite{Moslen-2016TWC,Du-2014ICL}, our scheme uses only the predictable large-scale CSI between UAVs and the legitimate receivers/eavesdroppers.

The transmission energy constraint at each UAV is another challenge for the secrecy performance of UAV swarm-enabled aerial networks. Since it's generally difficult to recharge the battery of the UAV during its flight, not only the transmission power budget but also the total transmission energy constraint should be taken into account for each UAV. Note that the total energy consumption of the UAV consists of two components: the UAV’s transmission energy consumption, which is due to the radiation, signal processing as well as other circuitry, and the UAV’s propulsion energy consumption, which is determined by the UAV's flying status including the velocity and acceleration \cite{Zeng-TWC2017}. Due to the signal processing of the confidential messages and the circuitry, this work mainly focuses on the transmission energy constraint per UAV. 

In addition, the transmission duration used to serve the legitimate users also plays a vital role to improve the secrecy performance of the system. Due to the limited total transmission energy at each UAV, the transmission duration may be restricted, which should be well designed for each legitimate user.

Motivated by the above observations, we investigate the AN-aided secure transmission for the UAV swarm-enabled aerial CoMP, where both of the legitimate receivers and eavesdroppers are equipped with multiple antennas. Different from the conventional eavesdropping, we assume the eavesdropper randomly walks following the trajectory of the UAV swarm, which may significantly deteriorate the secrecy performance of the system. In addition, unlike the existing AN-aided secure transmission based on the instantaneous CSI of the legitimate users, AN in our proposed scheme is designed by using only the large-scale CSI. To the best of our knowledge, this is the first time that studies AN-assisted secure transmission in the UAV swarm-enabled aerial CoMP by exploiting only the large-scale CSI of the legitimate receivers/eavesdroppers. 

Our main contributions of the paper are summarized as follows:

\begin{itemize}
	\item We consider physical layer security in the UAV swarm-enabled aerial networks. Specifically, multiple single-antenna rotary-wing UAVs perform an aerial CoMP, and enable a virtual MIMO transmission link with the multiple-antenna legitimate receivers or the eavesdropper, in which the swarm transmits the confidential messages in conjunction with AN and sequentially hovers to serve the scheduled legitimate users. Unlike the existing wiretap mode where the eavesdropper keeps static at a fixed location, we consider the eavesdropper moves following the trajectory of the swarm for better eavesdropping in a passive manner.
	
	\item To characterize the typical propagation environments, we consider a practical composite channel model consisting of both small-scale and large-scale fading. However, it is infeasible to achieve perfect CSI since the small-scale channel fading is time-varying and hard to be acquired. In this work, we use only the large-scale channel fading, which is more reasonable because the large-scale channel fading mainly depends on the position information of both the UAV and the legitimate receivers/eavesdroppers. We can obtain such information based on the historical data and the related distance between the UAV and the legitimate receivers/eavesdroppers.
	
	\item Based on the large-scale CSI, an optimization framework in a whole-trajectory-oriented manner is proposed to maximize the secrecy throughput by jointly optimizing the power allocation between the confidential messages and AN as well as the transmission durations of all the legitimate receivers subject to the overall transmission energy constraint at each rotary-wing UAV. The formulated problem is not convex and hard to be solved directly. To deal with that, an equivalent max-min problem is reformulated by adopting the random matrix theory, and then an efficient iterative algorithm is proposed. Specifically, the problem is split into four subproblems. For the first two and the fourth subproblems, they are convex and can be solved using the general optimization toolbox. For the third subproblem, we first transform its non-convex behavior into the convex one by adopting a successive convex approximation technique. Then, these four subproblems are alternately updated in each iteration. Furthermore, we show that the proposed algorithm guarantees the convergence. Finally, simulation results validate that our proposed scheme could achieve a good secrecy performance. 
\end{itemize}

\subsection{Organization and Notations}
The rest of the paper is organized as follows. Section II presents the system model and problem formulation. Section III proposes power allocation for secure aerial CoMP. In Section IV, simulation results and discussions are presented. Finally, conclusions are made in Section VI.

Throughout this paper, upper case and lower case boldface letters represent the matrices and the vectors, respectively. $\mathbf{I}_L$ is an $L\times L$ identity matrix, and $\mathbf{0}$ is a zero vector. $\mathbb{E}(\cdot)$ denotes the expectation operation. $(\cdot)^H$ and $\text{Tr}(\cdot)$ represent the conjugate transpose and the trace of a matrix, respectively. $\mathbf{A} \succeq \mathbf{0}$ denotes that $\mathbf{A}$ is a positive semidefinite matrix. $y \sim \mathcal{N}(0,a)$ denotes the Gaussian random variable with mean 0 and variance $a$. $\mathbf{x} \sim \mathcal{CN}(\mathbf{s},\mathbf{\Sigma})$ is the complex circularly symmetric Gaussian distribution with the mean vector $\mathbf{s}$ and the covariance matrix $\mathbf{\Sigma}$.

	\section{System Model and Problem Formulation}
	\subsection{System Model}
	We consider the downlink transmission in the UAV swarm-enabled aerial networks. As illustrated in Fig. \ref{System_model}, the system consists of $L$ rotary-wing UAVs (indexed with $1,...,L$), $N$  legitimate users (indexed with $1,...,N$) as Bob, and one eavesdropper as Eve. All the legitimate users and the eavesdropper are equipped with $N_B$ and $N_E$ antennas, respectively. For the rotary-wing UAVs, they form a UAV swarm via CoMP, and act as an aerial base station to assist the wireless networks. When the UAVs serve the legitimate users, they are randomly dispatched in a circle centering at the legitimate users with an altitude. 
	
	Due to the limited weight and size, only one single antenna is equipped at each UAV. Over the flight of the UAV swarm, the legitimate users could be provided with the confidential messages when being scheduled \footnote{In fact, the user scheduling issue is important for the system. However, it is out of the scope of this work.}. Due to the openness of the wireless link, there exists a leakage of the confidential messages. In this system, we assume the eavesdropper is passive and only intends for the confidential messages which are transmitted to the scheduled legitimate users. Furthermore, the eavesdropper randomly moves following the specific trajectory of the swarm to improve eavesdropping. Meanwhile, the eavesdropper also tries to keep a safety distance $r_e$ away from the scheduled legitimate users so that it could not be spotted.
	
	\begin{figure*}[tb]
		\centering       		 			 
		\includegraphics[scale=0.9]{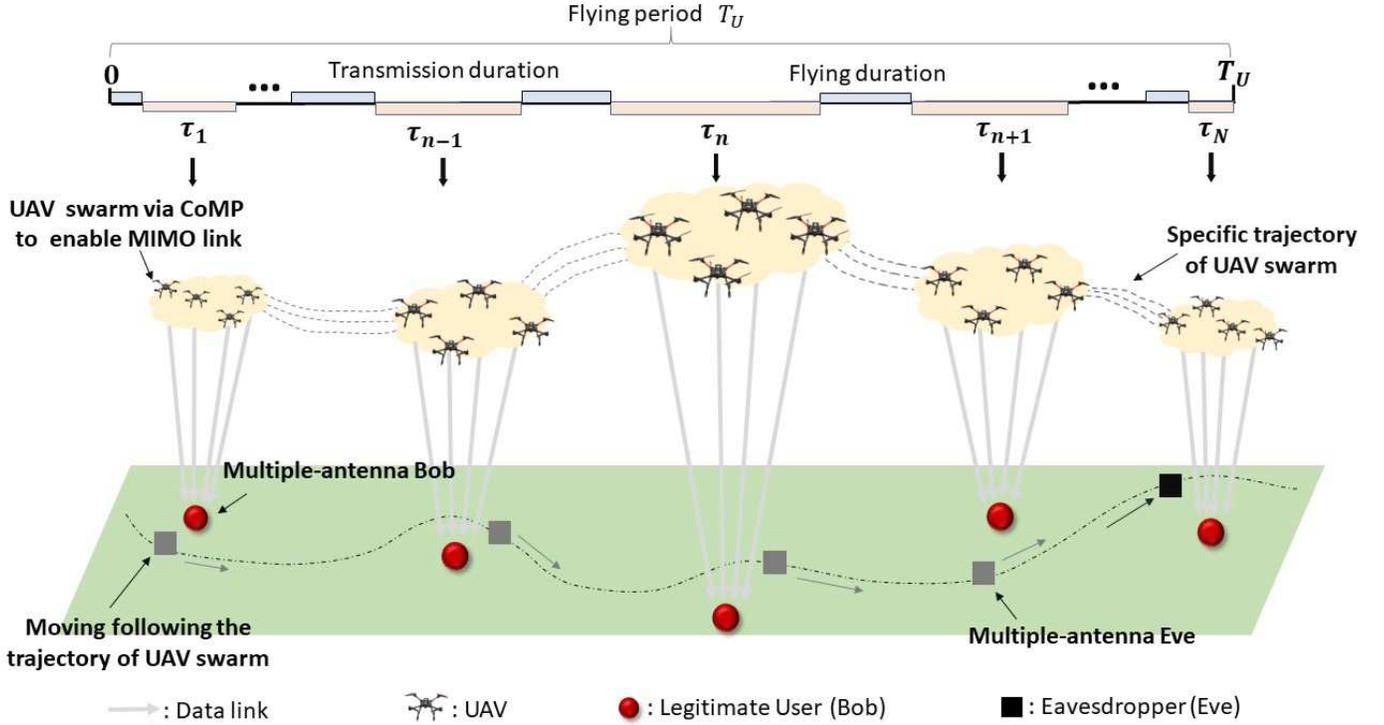}       	 
		\caption{Illustration of a UAV swarm-enabled aerial network, where a UAV swarm, acting as an aerial CoMP, enables MIMO secure communications with the multiple-antenna legitimate users and the eavesdropper in a whole-trajectory-oriented manner. Specifically, the swarm hovers to serve the scheduled legitimate user in the transmission duration, and then it flies to the next scheduled one within a flying duration. For the eavesdropper, it wiretaps the confidential messages by moving following the trajectory of the UAV swarm.} 
		\label{System_model}
	\end{figure*} 
%
	
	In this work, we assume the same consecutive period, denoting as $T_U$, for each rotary-wing UAV in the swarm, which mainly consists of two parts: the total flying duration $T_f$ and the total transmission duration $T_t$. $T_f$ is highly dependent on the velocity of the UAVs, and $T_t$ is determined by the transmission energy consumption of the UAVs. In this work, the UAVs are assumed to serve the scheduled legitimate receivers in a sequential manner. Specifically, they hover above a scheduled legitimate receiver to send the confidential messages during a transmission duration, and then fly to the next scheduled one within a flying duration. In addition, each legitimate receiver is assumed to be scheduled at most once during the total transmission duration. Denote the transmission durations for the scheduled legitimate users as $\tau_1$, $\tau_2$, ...,$\tau_N$, respectively, which change in practice and need to satisfy
	\begin{equation}
	\sum_{n=1}^{N}\tau_n \leq T_t,
	\end{equation}
	\begin{equation}
	0\leq \tau_n \leq \tau^{\max},\forall n,
	\end{equation}
	where $\tau^{\max}$ denotes the transmission duration budget. Note that the transmission duration for each scheduled legitimate user is normally short due to the limited transmission energy of the rotary-wing UAV \cite{Feng-2019IoTJournal}.
	
	Suppose that the coordinate of the $l$th UAV in the $n$th transmission duration is $(w_l[n],s_l[n],h_l[n])$ \footnote{The UAVs fly following a specific trajectory, whose locations can be achieved by using the existing positioning technology, i.e., the technology of Global Positioning System (GPS) and BeiDou Navigation Satellite System (BDS) combined positioning \cite{CSNC2016-Wang}. Based on the open channel link, the UAVs can broadcast their locations to legitimate receivers for them to follow.}, where $(w_l[n],s_l[n])$ and $h_l[n]$ denote the horizontal coordinate and the altitude of the $l$th UAV, respectively. Denote the coordinate of the scheduled legitimate user/eavesdropper in the $n$th transmission duration as $(r_q[n],t_q[n],0)$, where $q \in \{B, E\}$.  Thus, the distance between the $l$th UAV and $q$ at the $n$th transmission duration is 
	\begin{equation}
	\small
	\begin{aligned}
	&d_{q,l}[n]\\
	=& \sqrt{(h_l[n])^2+\big(w_l[n]-r_q[n]\big)^2+\big(s_l[n]-t_q[n]\big)^2}, q\in \{B,E\}.
	\end{aligned}
	\end{equation}
	
	To be practical, we consider both line-of-sight (LoS) and non-line-of-sight (NLoS) connections between the UAVs and the legitimate users. Therefore, the large-scale path loss between the $l$th UAV and $q$ at the $n$th transmission duration can be modeled as \cite{Al-Hourani-2014WCL}
	\begin{equation}
	\text{PL}_{q,l}^{dB}[n]=\frac{A}{1+ae^{-b(\rho_{q,l}[n]-a)}}+B_{q,l}[n],
	\end{equation}
	where
	\begin{equation}
	\allowdisplaybreaks
	\nonumber
	\begin{aligned}
	& A=\eta_{\text{LoS}}-\eta_{\text{NLoS}},\\
	&B_{q,l}[n]=20\text{lg}(d_{q,l}[n])+20\text{lg}(\frac{4\pi f}{c})+\eta_{\text{NLoS}},\\
	& \rho_{q,l}[n]=\frac{180}{\pi}\arcsin\Bigg(\frac{h_l[n]}{d_{q,l}[n]}\Bigg),
	\end{aligned}
	\end{equation} 
	$\eta_{\text{LoS}}$, $\eta_{\text{NLoS}}$, $a$ and $b$ are constants related to the propagation environment, $f$ is the carrier frequency, and $c$ is the speed of light \cite{Al-Hourani-2014WCL}.
	
	Consequently, the absolute power loss between the $l$th UAV and $q$ at the $n$th transmission duration can be expressed as:
	\begin{equation}
	\label{absolute_power_loss}
	Q_{q,l}[n]=10^{\frac{\text{PL}_{q,l}^{dB}[n]}{10}}.
	\end{equation} 
	
	The channel from the $l$th UAV to $q$ at the $n$th transmission duration can be rewritten as
	\begin{equation}
	\mathbf{h}_{q,l}[n]=Q_{q,l}^{-\frac{1}{2}}[n]\mathbf{s}_{q,l}[n],
	\end{equation}
	where $\mathbf{s}_{q,l}[n] \in \mathbb{C}^{N_q\times 1}$ represents the small-scale fading between the $l$th UAV and $q$, of which the entries are independently and identically distributed (i.i.d) according to $\mathcal{CN}(0,1)$.
	
	In order to degrade the eavesdropper's channel, each UAV transmits the confidential message in conjunction with AN. Denoting $x_{l}[n]$ as the transmission signal from the $l$th UAV to the scheduled legitimate user at the $n$th transmission duration, we have  
	\begin{equation}
	x_{l}[n]=x_{l}^{s}[n]+x_{l}^{a}[n],
	\end{equation}
	where $x_l^s[n]$ and $x_l^a[n]$ represent the confidential message and AN from the $l$th UAV, respectively.
	
	Furthermore, we express the transmission power from the $l$th UAV to the scheduled legitimate user at the $n$th transmission duration as
	\begin{equation}
	\mathbb{E}\{\rvert x_{l}[n] \rvert^2\}=p_{l}^{s}[n]+p_{l}^{a}[n],
	\end{equation}
	where $\mathbb{E}\{\rvert x_{l}^{s}[n] \rvert^2\}=p_{l}^s[n]$, $\mathbb{E}\{\rvert x_{l}^{a}[n] \rvert^2\}=p_{l}^a[n]$, $p_{l}^s[n]$ and $p_{l}^a[n]$ denote the power of the confidential message and that of the artificial noise transmitted by the $l$th UAV for the scheduled legitimate user at the $n$th transmission duration, respectively.
	
	Since each UAV has the limited transmission power, we have 
	\begin{equation}
	\label{kd_peak_power}
	0 \leq p_{l}^s[n]+p_{l}^a[n]\leq P^{\max},\ \forall l,n
	\end{equation}
	where \eqref{kd_peak_power} represents the transmission power constraint and $P^{\max}$ is the transmission power budget of each UAV. Considering the short transmission duration, the transmission power is assumed to keep constant at the $n$th transmission duration.
	
	Considering the transmission energy limitation of the UAVs within the flying period, the following constraint is achieved 
	\begin{equation}
	\label{kd_average_peak_power}
	\sum_{n=1}^{N}(p_{l}^s[n]+p_{l}^a[n])\tau_n\leq E^{\max}, \forall l
	\end{equation}
	where \eqref{kd_average_peak_power} denotes the total transmission energy constraint at each UAV over the whole flight, and $E^{\max}$ is the transmission energy budget per UAV.
	
	Based on the aforementioned analysis, all the UAVs work together to transmit the confidential messages for the legitimate users, which could form a virtual $ N_q \times L $ MIMO communication link. Note that to avoid the collision, we assume the UAVs are restricted to fly following their specific trajectory with a minimum safety distance between them. In this case, the composite channel matrix $\mathbf{H}_q[n]\in \mathbb{C}^{{N_q}\times L}$ between the swarm and $q$ at the $n$th transmission duration can be expressed as
	\begin{equation}
	\mathbf{H}_q[n]=\mathbf{S}_q[n]\mathbf{Q}_q[n], q \in  \{B, E\}
	\end{equation} 
	where 
	\begin{equation}
	\nonumber
	\begin{aligned}
	&
	\mathbf{H}_q[n]=\big[\mathbf{h}_{q,1}[n],\mathbf{h}_{q,2}[n],...,\mathbf{h}_{q,L}[n]\big],\\
	\end{aligned}
	\end{equation}
	\begin{equation}
	\nonumber
	\mathbf{S}_q[n]=\big[\mathbf{s}_{q,1}[n],\mathbf{s}_{q,2}[n],...,\mathbf{s}_{q,L}[n]\big],
	\end{equation}
	\begin{equation}
	\nonumber
	\mathbf{Q}_q[n]=
	\begin{bmatrix}
	Q_{q,1}^{-\frac{1}{2}}[n] &  &\\
	&\ddots & \\
	& & Q_{q,L}^{-\frac{1}{2}}[n]
	\end{bmatrix}.
	\end{equation}
	
	The received signal at the scheduled legitimate user, denoting $\mathbf{y}_{B}[n]$, in the $n$th transmission duration, and that at the corresponding eavesdropper, denoting $\mathbf{y}_{E}[n]$, in the $n$th transmission duration are given by
	\begin{equation}
	\label{received_signal_B}
	\mathbf{y}_{B}[n]=\mathbf{H}_{B}[n]\big(\mathbf{x}_{s}[n]+\mathbf{x}_{a}[n]\big)+\mathbf{n}_{B}[n],
	\end{equation}
	\begin{equation}
	\label{received_signal_E}
	\mathbf{y}_{E}[n]=\mathbf{H}_{E}[n]\big(\mathbf{x}_{s}[n]+\mathbf{x}_{a}[n]\big)+\mathbf{n}_{E}[n],
	\end{equation}
	respectively, where $\mathbf{x}_{s}[n] \sim \mathcal{CN}(\mathbf{0},\mathbf{P}_{s}[n])$ and $\mathbf{x}_{a}[n] \sim \mathcal{CN}(\mathbf{0},\mathbf{P}_{a}[n])$ denote the confidential messages and AN transmitted by the UAV swarm at the $n$th transmission duration \footnote{Different from the existing literatures \cite{2008Goel-TWC, 2016TVT-Na}, AN in this work is designed by using only the large-scale CSI instead of the instantaneous legitimate CSI. Owing to the broadcast nature of the wireless channel, AN unavoidably has a leakage and harms the legitimate receivers. Thus, it's important to carefully design the power allocation between the confidential messages and AN so as to minimize the harmful effect on the legitimate users while jamming the eavesdropper, which would be presented in details in the following.}, respectively, and $\mathbf{P}_{s}[n]$ and $\mathbf{P}_{a}[n]$ are their covariance matrices, respectively. $\mathbf{n}_{B}[n] \sim \mathcal{CN}(\mathbf{0},\delta^2\mathbf{I}_{N_B})$ and $\mathbf{n}_{E}[n] \sim \mathcal{CN}(\mathbf{0},\delta^2\mathbf{I}_{N_E})$ denote the noise vectors at the scheduled legitimate user and the eavesdropper at the $n$th transmission duration, respectively, and $\delta^2$ represents the noise variance \footnote{Here, we assume the noise variance is the same, i.e., equal to $\delta^2$, over the flying period. For convenience, we drop $n$ here.}.
	\begin{remark}
		Note that the locations of the legitimate users/eavesdropper are \textit{prior} known by the UAV swarm for transmission resource allocation. Specifically, the positions of the legitimate receivers in different transmission durations could be detected by using GPS or Light Detection and Ranging (LiDAR) at the UAVs \cite{ISPRS2015}.
		
		Considering the eavesdropper's passive behavior when it randomly follows the specific trajectory of the UAV swarm, it may be difficult to achieve the accurate positions of the eavesdropper at the UAVs. In this work, we consider a worst case scenario, where the power allocation strategy is performed on the assumption that the eavesdropper locates at the optimal position with the strongest receiving power in a circle of radius $r_e$. Fig. \ref{top_view} illustrates the optimal location of the eavesdropper during the $n$th transmission duration. Specifically, for the point in the located circle of the eavesdropper with the central angle $\theta[n]$, the receiving power from the $l$th UAV is $\frac{p_l^s[n]+p_l^a[n]}{\bar{Q}_{E,l}[n]}\|\mathbf{s}_{E,l}[n]\|^2$, where $\bar{Q}_{E,l}[n]=Q_{E,l}[n]\big\rvert_{r_E[n]=r_B[n]+r_e\cos(\theta[n]),t_E[n]=t_B[n]+r_e\sin(\theta[n])}$ based on \eqref{absolute_power_loss}. Considering the constant transmission power in the $n$th transmission duration, the optimal central angle can be expressed as $\tilde{\theta}[n]=\arg\min\limits_{\theta[n]\in[0,2\pi]}\frac{\bar{Q}_{E,1}[n]+...+\bar{Q}_{E,L}[n]}{L}$. Then, the eavesdropper's optimal location is obtained as $(r_E[n],t_E[n],0)=(r_B[n]+r_e\cos(\tilde{\theta}[n]),t_B[n]+r_e\sin(\tilde{\theta}[n]),0)$.
	\end{remark}
	
		\begin{figure}[t!]
			\centering       		 			 
			\includegraphics[width= 0.8\columnwidth]{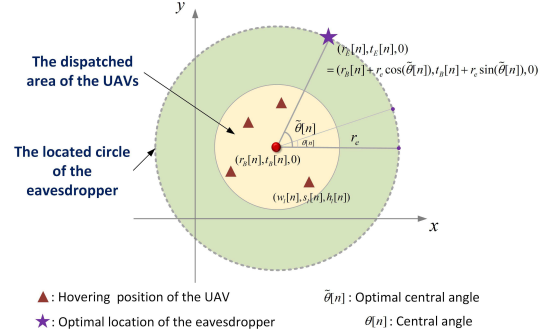}       	 
			\caption{Illustration of the eavesdropper's optimal location in the located circle of the eavesdropper with the safety distance during the $n$th transmission duration (top view of Fig. \ref{System_model}).} 
			\label{top_view}
		\end{figure} 
	\vspace{-0.5cm}
	\subsection{Problem Formulation}
	In this subsection, we focus on the problem formulation for this system. Based on \eqref{received_signal_B}, the achievable ergodic rate for the scheduled legitimate user at the $n$th transmission duration is given by 
	\begin{equation}
	\allowdisplaybreaks
	\label{R_B}
	\begin{aligned}
	R_{B}[n]=&\mathbb{E}_{\mathbf{S}_{B}[n]}\Bigg[\log_2 \det \Bigg(\mathbf{I}_{N_B}+\mathbf{H}_{B}[n]\mathbf{P}_{s}[n](\mathbf{H}_{B}[n])^H \\
	&\times\Big(\mathbf{H}_{B}[n]\mathbf{P}_{a}[n](\mathbf{H}_{B}[n])^H+\delta^2\mathbf{I}_{N_B}\Big)^{-1}\Bigg)\Bigg],
	\end{aligned}   	
	\end{equation}
	where $\mathbb{E}_{\mathbf{S}_{B}[n]}(\cdot)$ is taken over the random small-scale fading realization of $\mathbf{S}_{B}[n]$, 
	\begin{equation}
	\mathbf{P}_{s}[n]=
	\begin{bmatrix}
	p_{1}^s[n]& &\\
	& \ddots &\\
	& & p_{L}^s[n]\\
	\end{bmatrix},
	\end{equation}  
	and
	\begin{equation}
	\mathbf{P}_{a}[n]=
	\begin{bmatrix}
	p_{1}^a[n]& &\\
	& \ddots &\\
	& & p_{L}^a[n]\\
	\end{bmatrix}.
	\end{equation} 
	
	Based on \eqref{received_signal_E}, the achievable ergodic rate for the eavesdropper who is intended for the confidential message of the scheduled legitimate user at the $n$th transmission duration is 
	\begin{equation}
	\allowdisplaybreaks
	\label{R_E}
	\begin{aligned}
	R_{E}[n]=&\mathbb{E}_{\mathbf{S}_{E}[n]}\Bigg[\log_2 \det \Bigg(\mathbf{I}_{N_E}+\mathbf{H}_{E}[n]\mathbf{P}_{s}[n](\mathbf{H}_{E}[n])^H \\
	&\times\Big(\mathbf{H}_{E}[n]\mathbf{P}_{a}[n](\mathbf{H}_{E}[n])^H+\delta^2\mathbf{I}_{N_E}\Big)^{-1}\Bigg)\Bigg],
	\end{aligned}   	
	\end{equation} 
	where $\mathbb{E}_{\mathbf{S}_{E}[n]}(\cdot)$ is taken over the random small-scale fading realization of  $\mathbf{S}_{E}[n]$.
	
	Denote $\mathbf{\Xi}=\{\tau_n, \forall n\}$, $\mathbf{\Phi}_s=\{\mathbf{P}_s[n],\forall n\}$ and $\mathbf{\Phi}_a=\{\mathbf{P}_a[n],\forall n\}$. Then, the secrecy throughput for the UAV swarm-enabled aerial networks is defined as \cite{Chu-TGC2018}
	\begin{equation}
	\label{R}
	R\big(\mathbf{\Xi},\mathbf{\Phi}_{s},\mathbf{\Phi}_{a}\big)=\frac{1}{T_U}\sum_{n=1}^{N}\tau_n\Big[R_{B}[n]-R_{E}[n]\Big]^+,
	\end{equation}
	where $[x]^+=\max(0,x)$.
	
	In this work, our goal is to maximize the secrecy throughput over the flying period of the UAV swarm by jointly optimizing the power of the confidential messages (i.e., $\mathbf{\Phi}_{s}$) and AN power (i.e., $\mathbf{\Phi}_{a}$) as well as the transmission durations of the legitimate users (i.e., $\mathbf{\Xi}$) under the constraint of the total transmission energy for each UAV. The optimization problem can be formulated as
	\begin{subequations}
		\allowdisplaybreaks
		\label{original_problem}
		\begin{align}
		\allowdisplaybreaks
		\max_{\tiny \mathbf{\Xi}, \mathbf{\Phi}_{s},\mathbf{\Phi}_{a}} & \ \ \frac{1}{T_U}\sum_{n=1}^{N}\tau_n\Big[R_{B}[n]-R_{E}[n]\Big]^+ \label{op_objective1} \\
		\text{s.t.} \ \
		& 0 \leq p_{l}^s[n]+p_{l}^a[n]\leq P^{\max}, \forall l,n   & \label{power_constraint1}\\
		& \sum_{n=1}^{N}(p_{l}^s[n]+p_{l}^a[n])\tau_n\leq E^{\max}, \forall l  \label{energy_constraint1}\\
		& \sum_{n=1}^{N}\tau_n \leq T_t,  \label{total_tau_constraint1}\\
		&0\leq \tau_n\leq \tau^{\max}, \forall n \label{tau_n1}\\
		& \mathbf{P}_{s}[n]\succeq \mathbf{0}, \forall n, \label{P_s_constraint1} \\
		& \mathbf{P}_{a}[n]\succeq \mathbf{0}, \forall n.  \label{P_a_constraint1}
		\end{align}
	\end{subequations}    
	
	It can be observed that problem \eqref{original_problem} is challenging to be solved for two reasons. First, the operator $[\cdot]^+$ results in a nonsmooth manner. Second, even without $[\cdot]^+$, the objective function \eqref{op_objective1} has integrals with the expectation operator $\mathbb{E(\cdot)}$, which is intractable and difficult to achieve an explicit expression in terms of $\mathbf{\Phi}_s$, $\mathbf{\Phi}_a$ and $\mathbf{\Xi}$. 
	\begin{remark}
		It is worth mentioning that even though the UAV's propulsion energy is much higher than the UAV's transmission energy \cite{Zeng-TWC2017}, we mainly consider the total transmission energy constraint at each UAV in this work. The reasons are as follows:
		\begin{enumerate}[(1)]
			\item Based on \eqref{received_signal_B} - \eqref{R}, secrecy throughput is highly related to the UAV transmission power consumption. Thus, to maximize the secrecy throughput, the UAV transmission energy consumption is more important than the UAV’s propulsion energy consumption.
			\item We investigate the power allocation strategy to enhance the physical layer security of the downlink transmission in this work, which is part of the process of the signal processing. Based on problem \eqref{original_problem}, the power mainly denotes the transmission powers of the confidential messages and AN transmitted by the UAVs. Since the propulsion energy is mostly used to keep the UAV aloft as well as support its mobility, it nearly has little contribution to the improvement of the secrecy performance of the system.
			\item We perform the power allocation by considering an optimization framework in a whole-trajectory-oriented manner, where the UAVs follow a specific trajectory. In this case, the UAV's flying status, i.e., the flying period, velocity, propulsion energy consumption and so on, can be controlled with commands from the ground station when they serve the legitimate users/eavesdropper, which leads to the limited consideration of propulsion energy.  
		\end{enumerate}
	\end{remark}
	\vspace{-0.5cm}
\section{Power Allocation for Secure Aerial CoMP}
In this section, we devote our effort to achieve the optimal solutions of problem \eqref{original_problem}. Before the further analysis, we first handle the nonsmooth of the objective function in problem \eqref{original_problem} by adopting the similar analysis in \cite{Cui-TVT2018}. Then, \eqref{original_problem} can be reformulated into 
\begin{subequations}
	\label{original_problem1}
	\allowdisplaybreaks
	\begin{align}
	\max_{\tiny \mathbf{\Xi}, \mathbf{\Phi}_{s},\mathbf{\Phi}_{a}} & \ \ \frac{1}{T_U}\sum_{n=1}^{N}\tau_n\Big[R_{B}[n]-R_{E}[n]\Big] \label{op_objective} \\
	\text{s.t.} \ \
	& 0 \leq p_{l}^s[n]+p_{l}^a[n]\leq P^{\max}, \forall l,n   & \label{power_constraint}\\
	& \sum_{n=1}^{N}(p_{l}^s[n]+p_{l}^a[n])\tau_n\leq E^{\max}, \forall l  \label{energy_constraint}\\
	& \sum_{n=1}^{N}\tau_n \leq T_t, \label{total_tau_constraint}\\
	&  0\leq \tau_n\leq \tau^{\max}, \forall n \label{tau_n}\\
	& \mathbf{P}_{s}[n]\succeq \mathbf{0}, \forall n, \label{P_s_constraint} \\
	& \mathbf{P}_{a}[n]\succeq \mathbf{0}, \forall n,  \label{P_a_constraint}
	\end{align}
\end{subequations} 
where problem \eqref{original_problem1} and problem \eqref{original_problem} share the same optimal solution, and $ R\big(\mathbf{\Xi},\mathbf{\Phi}_{s},\mathbf{\Phi}_{a}\big)=\sum_{n=1}^{N}\frac{\tau_n}{T_U}\Big[R_{B}[n]-R_{E}[n]\Big]$.

To achieve the efficient power allocation, an explicit expression of the objective function in \eqref{original_problem1} is necessary. Although some works have provided an insightful result to obtain the analytical expression for the objective function, it is generally too cumbersome to do the further power allocation design since the analytical result involves a series of integrals \cite{Hanlen-2012TIT}.

In the following, we first achieve the closed form of the secrecy throughput in terms of $\mathbf{\Phi}_s$,  $\mathbf{\Phi}_a$ and $\mathbf{\Xi}$ by removing the expectation operator $\mathbb{E}(\cdot)$ based on \cite{Feng-2013}. Then, we reformulate the optimization problem. Finally, a computationally efficient iterative algorithm is proposed for the problem and its convergence is presented.

\subsection{Problem Transformation}	

Let $\mathbf{P}_u[n]=\mathbf{P}_s[n]+\mathbf{P}_a[n], \forall n$, where $\mathbf{P}_u[n]=\text{diag}\big[p_1^u[1],... ,p_L^u[N]\big]$, and $\mathbf{\Phi}_u=\{\mathbf{P}_u[n],\forall n\}$. Then, the secrecy throughput $ R\big(\mathbf{\Xi}, \mathbf{\Phi}_{s},\mathbf{\Phi}_{a}\big)$ can be equivalently rewritten as
\begin{equation}
\allowdisplaybreaks
\small
\label{eqn}
\begin{aligned}
&R\big(\mathbf{\Xi}, \mathbf{\Phi}_{u},\mathbf{\Phi}_{a}\big) \\
=&\frac{1}{T_U}\sum_{n=1}^{N}\tau_n
\Bigg[ 
\mathbb{E}_{\mathbf{S}_{B}[n]}\Big[\log_2 \det \Big(\mathbf{I}_{N_B}+\frac{\mathbf{H}_{B}[n]\mathbf{P}_{u}[n](\mathbf{H}_{B}[n])^H}{\delta^2}\Big)\Big]\\
&\ \ \ \ \ \ \ \ \ \ \ \   - \mathbb{E}_{\mathbf{S}_{B}[n]}\Big[\log_2 \det \Big(\mathbf{I}_{N_B}+\frac{\mathbf{H}_{B}[n]\mathbf{P}_{a}[n](\mathbf{H}_{B}[n])^H}{\delta^2} 
\Big)\Big]\\
& \ \ \ \ \ \ \ \ \ \ \ \  -\mathbb{E}_{\mathbf{S}_{E}[n]}\Big[\log_2 \det \Big(\mathbf{I}_{N_E}+\frac{\mathbf{H}_{E}[n]\mathbf{P}_{u}[n](\mathbf{H}_{E}[n])^H}{\delta^2}\Big)\Big]\\
& \ \ \ \ \ \ \ \ \ \ \ \  +\mathbb{E}_{\mathbf{S}_{E}[n]}\Big[\log_2 \det \Big(\mathbf{I}_{N_E}+\frac{\mathbf{H}_{E}[n]\mathbf{P}_{a}[n](\mathbf{H}_{E}[n])^H}{\delta^2} 
\Big)\Big]
\Bigg].
\end{aligned}
\end{equation}
It can be observed that \eqref{eqn} is still intractable due to the expectation operator $\mathbb{E}(\cdot)$. To cope with that, we try to approximate \eqref{eqn} by introducing the following theorem.
\begin{figure} [t!]  
	\centering 		 		 
	\includegraphics[width=0.9\columnwidth]{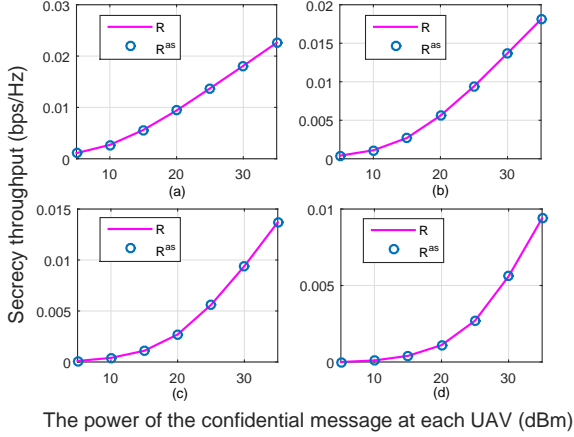}       	 
	\caption{Illustration of the approximate accuracy of $R^{as}\big(\mathbf{\Xi}, \mathbf{\Phi}_{u},\mathbf{\Phi}_{a}\big)$ with the power of the confidential message at each UAV, where we set $N_B=5$, $N_E=3$ and $L=7$. We specify $\mathbf{P}_a[n]=p_a\mathbf{I}_L, \forall n$ as the input AN power. Particularly, $p_a$ from (a) - (d) in the figure are set $5$dBm, $10$dBm, $15$dBm, $20$dBm, respectively.} 
	\label{R-Ra}
\end{figure}
\begin{theorem}
	\label{theorem_for_Rap}
	By introducing auxiliary variables $\mathbf{t}_{B,u}=\{{t_{B,u}[n], \forall n}\}$, $\mathbf{t}_{E,a}=\{{t_{E,a}[n], \forall n}\}$, $\mathbf{t}_{B,a}=\{{t_{B,a}[n], \forall n}\}$, $\mathbf{t}_{E,u}=\{{t_{E,u}[n], \forall n}\}$, $R\big(\mathbf{\Xi}, \mathbf{\Phi}_{u},\mathbf{\Phi}_{a}\big) $ in \eqref{eqn} can be equivalently expressed as
	\begin{equation}
	\small
	\label{R_ap}
	\begin{aligned}
	&R^{as}\big(\mathbf{\Xi}, \mathbf{\Phi}_{u},\mathbf{\Phi}_{a}\big)\\ 
	=&
	\begin{aligned}
	\max_{\tiny
		\mathbf{t}_{B,a}\geq \mathbf{0},
		\mathbf{t}_{E,u}\geq \mathbf{0},}
	\end{aligned} 
	\begin{aligned} \
	\min_{\tiny
		\mathbf{t}_{B,u}\geq \mathbf{0},
		\mathbf{t}_{E,a}\geq \mathbf{0},}
	\end{aligned} 
	\ \  \mathcal{G}(\mathbf{\Xi}, \mathbf{\Phi}_{u},\mathbf{\Phi}_{a},\mathbf{t}_{B,u},\mathbf{t}_{B,a},\\
	& \ \  \ \ \ \ \ \ \ \ \ \ \ \ \ \ \ \ \ \ \ \ \ \ \ \ \ \ \ \ \ \ \ \ \ \ \ \ \ \ \ \ \ \ \ \  \ \ \ \ \ \ \ \ \ \ \mathbf{t}_{E,u},\mathbf{t}
	_{E,a}),
	\end{aligned}
	\end{equation}
	where $\mathcal{G}(\mathbf{\Xi}, \mathbf{\Phi}_{u},\mathbf{\Phi}_{a},\mathbf{t}_{B,u},\mathbf{t}_{B,a},\mathbf{t}_{E,u},\mathbf{t}
	_{E,a})$ is defined in \eqref{eqx_approximate} on the top of the next page. Furthermore, $\mathcal{G}(\mathbf{\Xi}, \mathbf{\Phi}_{u},\mathbf{\Phi}_{a},\mathbf{t}_{B,u},\mathbf{t}_{B,a},\mathbf{t}_{E,u},\mathbf{t}
	_{E,a})$ is convex in terms of ($\mathbf{t}_{B,u}$, $\mathbf{t}_{E,a}$), and concave with respect to ($\mathbf{t}_{B,a}$, $\mathbf{t}_{E,u}$). 
\end{theorem} 
\begin{Proof}
	Please refer to Appendix \ref{appendix_theorem_Rap}.
\end{Proof}
\begin{figure*}[ht]
	
	\normalsize
	\begin{equation}
	\label{eqx_approximate}
	\begin{aligned}
	&\mathcal{G}(\mathbf{\Xi}, \mathbf{\Phi}_{u},\mathbf{\Phi}_{a},\mathbf{t}_{B,u},\mathbf{t}_{B,a},\mathbf{t}_{E,u},\mathbf{t}
	_{E,a})\\
	=&\frac{1}{T_U}\sum_{n=1}^{N}\tau_n\Bigg[g\Big(\mathbf{P}_u[n],N_B,\mathbf{Q}_B[n],t_{B,u}[n]\Big)
	-g\Big(\mathbf{P}_a[n],N_B,\mathbf{Q}_B[n],t_{B,a}[n]\Big)\\
	&\ \ \ \ \ \ \ \ \ \  \ \ \ \ \ - g\Big(\mathbf{P}_u[n],N_E,\mathbf{Q}_E[n],t_{E,u}[n]\Big)
	+g\Big(\mathbf{P}_a[n],N_E,\mathbf{Q}_E[n],t_{E,a}[n]\Big)\Bigg]\\
	=&\frac{1}{T_U}\sum_{n=1}^{N}\tau_n\Bigg\{\Big[\log_2 \det\Big(\mathbf{I}_L+\frac{N_B\mathbf{Q}_B[n]\mathbf{P}_u[n]\mathbf{Q}_B[n]}{\delta^2e^{t_{B,u}[n]}}\Big) +  N_B\log_2e\Big(t_{B,u}[n]-1+e^{-t_{B,u}[n]}\Big)\Big]\\
	&\ \ \ \ \ \ \ \ \ \  \ \ \ \ \ \ -\Big[\log_2 \det\Big(\mathbf{I}_L+\frac{N_B\mathbf{Q}_B[n]\mathbf{P}_a[n]\mathbf{Q}_B[n]}{\delta^2e^{t_{B,a}[n]}}\Big)+  N_B\log_2e\Big(t_{B,a}[n]-1+e^{-t_{B,a}[n]}\Big)\Big]\\
	&\ \ \ \ \ \ \ \ \ \  \ \ \ \ \ \ -\Big[\log_2 \det\Big(\mathbf{I}_L+\frac{N_E\mathbf{Q}_E[n]\mathbf{P}_u[n]\mathbf{Q}_E[n]}{\delta^2e^{t_{E,u}[n]}}\Big) +  N_E\log_2e\Big(t_{E,u}[n]-1+e^{-t_{E,u}[n]}\Big)\Big]\\
	&\ \ \ \ \ \ \ \ \ \  \ \ \ \ \ \  +\Big[\log_2 \det\Big(\mathbf{I}_L+\frac{N_E\mathbf{Q}_E[n]\mathbf{P}_a[n]\mathbf{Q}_E[n]}{\delta^2e^{t_{E,a}[n]}}\Big)
	+  N_E\log_2e\Big(t_{E,a}[n]-1+e^{-t_{E,a}[n]}\Big)\Big] \Bigg\}. 			
	\end{aligned}
	\end{equation}
	\hrulefill 
	\vspace*{4pt}
\end{figure*}


Note that the performance gap between $R\big(\mathbf{\mathbf{\Xi}, \Phi}_{u},\mathbf{\Phi}_{a}\big)$ in \eqref{eqn} and $R^{as}\big(\mathbf{\Xi}, \mathbf{\Phi}_{u},\mathbf{\Phi}_{a}\big)$ in \eqref{R_ap} is negligible, which has been demonstrated in Fig. \ref{R-Ra}. Without loss of generality, we randomly select a system topology, and adopt the values of simulation settings in Section IV with different transmission powers based on urban scenario. Specifically, we first verify that $\mathcal{K}(\mathbf{\Xi},\mathbf{\Phi}_u,\mathbf{\Phi}_a,\mathbf{w}_{B,u},\mathbf{w}_{B,a},\mathbf{w}_{E,u},\mathbf{w}_{E,a})=y(\mathbf{\Xi},\mathbf{\Phi}_u,\mathbf{w}_{B,u})-y(\mathbf{\Xi},\mathbf{\Phi}_a,\mathbf{w}_{B,a})-y(\mathbf{\Xi},\mathbf{\Phi}_u,\mathbf{w}_{E,u})+y(\mathbf{\Xi},\mathbf{\Phi}_a,\mathbf{w}_{E,a})$, which is based on the definition \eqref{approx_eq} in Appendix \ref{appendix_theorem_Rap}, is the closed form of $R\big(\mathbf{\Xi},\mathbf{\Phi}_{u},\mathbf{\Phi}_{a}\big)$ by simulation in Fig. \ref{R-Ra}. Then, considering the relationship between $\mathcal{K}(\mathbf{\Xi},\mathbf{\Phi}_u,\mathbf{\Phi}_a,\mathbf{w}_{B,u},\mathbf{w}_{B,a},\mathbf{w}_{E,u},$ $\mathbf{w}_{E,a})$ and $\mathcal{G}(\mathbf{\Xi}, \mathbf{\Phi}_{u},\mathbf{\Phi}_{a},$ $\mathbf{t}_{B,u},\mathbf{t}_{B,a},\mathbf{t}_{E,u},\mathbf{t}_{E,a})$ according to the derivation \eqref{max_min_equ} in Appendix \ref{appendix_theorem_Rap}, it is obvious to see that $R^{as}\big(\mathbf{\Xi}, \mathbf{\Phi}_{u},\mathbf{\Phi}_{a}\big)$ is a very accurate approximation for $R\big(\mathbf{\Xi},\mathbf{\Phi}_{u},\mathbf{\Phi}_{a}\big)$. In this case, the optimization problem \eqref{original_problem1} can be reformulated as
\begin{subequations}
	\allowdisplaybreaks
	\label{eqn_relax3}
	\begin{align}
	\max_{\tiny
		\begin{aligned}
		&\mathbf{\Xi},	\mathbf{\Phi}_{u},\mathbf{\Phi}_{a},\\
		&\mathbf{t}
		_{B,a},\mathbf{t}_{E,u}\\
		\end{aligned} }
	\min_{\tiny
		\begin{aligned}
		\mathbf{t}_{B,u},
		\mathbf{t}_{E,a}
		\end{aligned} }
	\ & 	\mathcal{G}(\mathbf{\Xi},\mathbf{\Phi}_{u},\mathbf{\Phi}_{a},\mathbf{t}_{B,u},\mathbf{t}_{B,a},\mathbf{t}_{E,u},\mathbf{t}
	_{E,a})\\
	\text{s.t.}\ \ \ \ \ \
	\   &0 \leq p_l^u[n] \leq P^{\max}, \forall l,n \label{Pu_constraint}\\
	& \sum_{n=1}^{N}p_l^u[n]\tau_n \leq E^{\max}, \forall l\\
	& \sum_{n=1}^{N}\tau_n \leq T_t,\\
	&  0\leq \tau_n\leq \tau^{\max}, \forall n \\
	& \mathbf{P}_{u}[n]-\mathbf{P}_a[n]\succeq \mathbf{0}, \forall n, \label{Pu_Pa_constraint}  \\
	& \mathbf{P}_{a}[n]\succeq \mathbf{0}, \forall n, \label{Pa_constraint}  \\
	& \mathbf{t}_{B,a}\geq \mathbf{0}, \mathbf{t}_{E,u}\geq \mathbf{0},\\
	& \mathbf{t}_{B,u}\geq \mathbf{0},\mathbf{t}_{E,a}\geq \mathbf{0}.
	\end{align}
\end{subequations}	

It's worth mentioning that by performing a serial of reformulations, problem \eqref{eqn_relax3} is equivalent to the original problem \eqref{original_problem}. In the following, we devote to an efficient optimization algorithm for the equivalent problem \eqref{eqn_relax3}.

\subsection{An Iterative Algorithm to Solve the Problem}

Note that since \eqref{eqn_relax3} is not convex, it's generally difficult to achieve the global optimal solutions. In this subsection, we propose an efficient iterative algorithm by adopting the block coordinate descent method to find the stationary solutions of \eqref{eqn_relax3}. We first decouple the optimization variables into the following four blocks $(\mathbf{t}_{B,u},\mathbf{t}_{E,a})$, $(\mathbf{t}_{B,a},\mathbf{t}_{E,u})$, $(\mathbf{\Phi}_u,\mathbf{\Phi}_a)$, $\mathbf{\Xi}$ and then alternately optimize these four blocks one by one by taking the other variables as the constants obtained in the last iteration. Specifically, for any given transmission power indicators ($\mathbf{\Phi}_u$, $\mathbf{\Phi}_a$) and the transmission duration $\mathbf{\Xi}$, the auxiliary variables ($\mathbf{t}_{B,u}$, $\mathbf{t}_{E,a}$) (or ($\mathbf{t}_{B,a}$, $\mathbf{t}_{E,u}$)) can be efficiently solved through standard algorithm \cite{Brayton-1979}. For any obtained auxiliary variables ($\mathbf{t}_{B,a}$, $\mathbf{t}_{E,u}$) and  ($\mathbf{t}_{B,u}$, $\mathbf{t}_{E,a}$) as well as the transmission duration $\mathbf{\Xi}$, the transmission power indicators ($\mathbf{\Phi}_u$, $\mathbf{\Phi}_a$) can be optimized by the successive convex approximation technique. Finally, the transmission duration $\mathbf{\Xi}$ can be achieved based on the obtained $(\mathbf{t}_{B,u},\mathbf{t}_{E,a})$, $(\mathbf{t}_{B,a},\mathbf{t}_{E,u})$ and $(\mathbf{\Phi}_u,\mathbf{\Phi}_a)$.

Denote $m\geq 1$ as the number of the iteration step. Problem \eqref{eqn_relax3} can be separated into the following four subproblems.

(1) The auxiliary variables $(\mathbf{t}_{B,u},\mathbf{t}_{E,a})$

In the $m$th iteration, we first optimize $(\mathbf{t}_{B,u},\mathbf{t}_{E,a})$ with $(\mathbf{\Phi}^{m-1}_u,\mathbf{\Phi}^{m-1}_a)$ and $\mathbf{\Xi}^{m-1}$ obtained in the $m-1$th iteration. In this case, problem \eqref{eqn_relax3} can be reformulated into 
\begin{subequations}
	\allowdisplaybreaks
	\label{subproblem1}
	\begin{align}
	\min_{\mathbf{t}_{B,u},\mathbf{t}_{E,a}} & \frac{1}{T_U}\sum_{n=1}^{N}\tau_n^{m-1}\Big[g\Big(\mathbf{P}^{m-1}_u[n],N_B,\mathbf{Q}_B[n],t_{B,u}[n]\Big)
	\\
	&\ \ \ \ \ \ \ \ \  \ \ \ \ \ \ \ \ \ \ \ \
	+g\Big(\mathbf{P}^{m-1}_a[n],N_E,\mathbf{Q}_E[n],t_{E,a}[n]\Big)\Big] \nonumber\\
	\text{s.t.} &\ \  \ \mathbf{t} _{B,u}\geq \mathbf{0}, \mathbf{t}_{E,a}\geq \mathbf{0}.
	\end{align}
\end{subequations}

Based on Theorem 1, we know that problem \eqref{subproblem1} is convex, which can be efficiently solved by means of the standard optimization toolbox, i.e., CVX.

(2) The auxiliary variables $(\mathbf{t}_{B,a},\mathbf{t}_{E,u})$

Based on the obtained $(\mathbf{\Phi}^{m-1}_u,\mathbf{\Phi}^{m-1}_a)$ and $\mathbf{\Xi}^{m-1}$ in the $m-1$th iteration, the variables $(\mathbf{t}_{B,a},\mathbf{t}_{E,u})$ can be achieved by 
\begin{subequations}
	\allowdisplaybreaks
	\label{subproblem2}
	\begin{align}
	\min_{\mathbf{t}_{B,a},\mathbf{t}_{E,u}} & \frac{1}{T_U}\sum_{n=1}^{N}\tau_n^{m-1}\Big[g\Big(\mathbf{P}^{m-1}_a[n],N_B,\mathbf{Q}_B[n],t_{B,a}[n]\Big)
	\\
	&\ \ \ \ \ \ \ \ \ \ \ \ \ \ \ \ \ \ \ \
	+g\Big(\mathbf{P}^{m-1}_u[n],N_E,\mathbf{Q}_E[n],t_{E,u}[n]\Big)\Big] \nonumber\\
	\text{s.t.} &\ \ \ \mathbf{t} _{B,a}\geq \mathbf{0}, \mathbf{t}_{E,u}\geq \mathbf{0}.
	\end{align}
\end{subequations}
which is convex and can be directly solved using CVX.

(3) The transmission power indicators $(\mathbf{\Phi}_u,\mathbf{\Phi}_a)$

The variables $(\mathbf{\Phi}_u,\mathbf{\Phi}_a)$ with the obtained $(\mathbf{t}^m_{B,u},\mathbf{t}^m_{E,a})$, $(\mathbf{t}^m_{B,a},\mathbf{t}^m_{E,u})$ and $\mathbf{\Xi}^{m-1}$ in the $m$th iteration can be achieved by solving the following problem 
\begin{subequations}
	\allowdisplaybreaks
	\label{subproblem3}
	\begin{align}
	\max_{\mathbf{\Phi}_u,\mathbf{\Phi}_a}  & \frac{1}{T_U}\sum_{n=1}^{N}\tau_n^{m-1}\Big[g\Big(\mathbf{P}_u[n],N_B,\mathbf{Q}_B[n],t^m_{B,u}[n]\Big)\\
	&\ \ \ \ \ \ \ \ \ \ \ \ \ \ \ \ \ \ \ \
	-g\Big(\mathbf{P}_a[n],N_B,\mathbf{Q}_B[n],t^m_{B,a}[n]\Big)\nonumber\\
	&\ \ \ \ \ \ \ \ \ \ \ \ \ \ \ \ \ \ \ \ -g\Big(\mathbf{P}_u[n],N_E,\mathbf{Q}_E[n],t^m_{E,u}[n]\Big) \nonumber\\
	&\ \ \ \ \ \ \ \ \ \ \ \ \ \ \ \ \ \ \ \
	+g\Big(\mathbf{P}_a[n],N_E,\mathbf{Q}_E[n],t^{m}_{E,a}[n]\Big)\Big]\nonumber\\	
	\text{s.t.} & \sum_{n=1}^{N}p_{l}^u[n]\tau_n^{m-1}\leq E^{\max}, \forall l \\
	&\eqref{Pu_constraint},\eqref{Pu_Pa_constraint}, \eqref{Pa_constraint}.\ \ 
	\end{align}
\end{subequations}

Note that both $g\big(\mathbf{P}_u[n],N_B,\mathbf{Q}_B[n],t_{B,u}^m[n]\big)$  and $g\big(\mathbf{P}_u[n],N_E,\mathbf{Q}_E[n],t_{E,u}^m[n]\big)$ are concave with respect to $\mathbf{P}_u[n]$, and both $g\big(\mathbf{P}_a[n],N_B,\mathbf{Q}_B[n],t_{B,a}^m[n]\big)$ and $g\big(\mathbf{P}_a[n],N_E,\mathbf{Q}_E[n],t_{E,a}^m[n]\big)$ are concave with respect to $\mathbf{P}_a[n]$ \cite{Feng-2013}. Thus, the objective function
in \eqref{subproblem3} actually mixes the addition and subtraction of these four concave terms, which is neither concave with respect to $\mathbf{\Phi}_{u}$ nor concave with respect to $\mathbf{\Phi}_{a}$. That is, problem \eqref{subproblem3} is not convex in terms of $\mathbf{\Phi}_{u}$ and $\mathbf{\Phi}_{a}$.

To overcome the convexity issue, we approximate the second and the third terms of the objective function in \eqref{subproblem3} to an affine function based on the first-order Taylor expansion.   

The gradient of $g\big(\mathbf{P}_a[n],N_B,\mathbf{Q}_B[n],t_{B,a}[n]\big)$\footnote{For convenience, we drop $m$.} can be rewritten as
\begin{equation}
\allowdisplaybreaks
\begin{aligned}
&\nabla_{\mathbf{P}_a[n]}\ g\big(\mathbf{P}_a[n],N_B,\mathbf{Q}_B[n],t_{B,a}[n]\big)\\
=&\varphi_{B,a}[n]\mathbf{Q}_B[n]\big(\mathbf{I}_L+\varphi_{B,a}[n]\mathbf{Q}_B[n]\mathbf{P}_a[n]\mathbf{Q}_B[n]\big)^{-1}\mathbf{Q}_B[n] \\
\triangleq& \mathcal{F}(\mathbf{P}_a[n],N_B,\mathbf{Q}_B[n],t_{B,a}[n]),
\end{aligned}
\end{equation}
where $\varphi_{B,a}[n]=\frac{\log_2eN_B}{\delta^2e^{t_{B,a}[n]}}$.
Thus, the first-order Taylor expansion of $g\big(\mathbf{P}_a[n],N_B,\mathbf{Q}_B[n],t_{B,a}[n]\big)$ at a certain point $\tilde{\mathbf{P}}_a[n]$ can be expressed as
\begin{equation}
\allowdisplaybreaks
\label{ga_Taylor}
\begin{aligned}
\allowdisplaybreaks
&g\big(\mathbf{P}_a[n],N_B,\mathbf{Q}_B[n],t_{B,a}[n]\big|\tilde{\mathbf{P}}_a[n]\big)\\
=&g\big(\tilde{\mathbf{P}}_a[n],N_B,\mathbf{Q}_B[n],t_{B,a}[n]\big)\\
&+
\text{tr}\big[\mathcal{F}\big(\tilde{\mathbf{P}}_a[n],N_B,\mathbf{Q}_B[n],t_{B,a}[n]\big)\big(\mathbf{P}_a[n]-\tilde{\mathbf{P}}_a[n]\big)\big],
\end{aligned}
\end{equation} 
where $g\big(\tilde{\mathbf{P}}_a[n],N_B,\mathbf{Q}_B[n],t_{B,a}[n]\big)=g\big(\mathbf{P}_a[n],N_B,$ $\mathbf{Q}_B[n],t_{B,a}[n]\big)|_{\mathbf{P}_a[n]=\tilde{\mathbf{P}}_a[n]}$. 
Clearly, $g\big(\mathbf{P}_a[n],N_B,\mathbf{Q}_B[n],$ $t_{B,a}[n]|\tilde{\mathbf{P}}_a[n]\big)$ is a linear function with respect to $\mathbf{P}_a[n]$. 

Similarly, the gradient of $g\big(\mathbf{P}_u[n],N_E,\mathbf{Q}_E[n],t_{E,u}[n]\big)$ can be rewritten as
\begin{equation}
\begin{aligned}
&\nabla_{\mathbf{P}_u[n]}\ g\big(\mathbf{P}_u[n],N_E,\mathbf{Q}_E[n],t_{E,u}[n]\big)\\
=&\varphi_{E,u}[n]\mathbf{Q}_E[n]\big(\mathbf{I}_L+\varphi_{E,u}[n]\mathbf{Q}_E[n]\mathbf{P}_u[n]\mathbf{Q}_E[n]\big)^{-1}\mathbf{Q}_E[n] \\
\triangleq&\mathcal{F}(\mathbf{P}_u[n],N_E,\mathbf{Q}_E[n],t_{E,u}[n]),
\end{aligned}
\end{equation}
where $\varphi_{E,u}[n]=\frac{\log_2eN_E}{\delta^2e^{t_{E,u}[n]}}$.
Thus, the first-order Taylor expansion of $g\big(\mathbf{P}_u[n],N_E,\mathbf{Q}_E[n],t_{E,u}[n]\big)$ at a certain point $\tilde{\mathbf{P}}_u[n]$ can be expressed as
\begin{equation}
\label{gu_Taylor}
\begin{aligned}
\allowdisplaybreaks
&g\big(\mathbf{P}_u[n],N_E,\mathbf{Q}_E[n],t_{E,u}[n]\big|\tilde{\mathbf{P}}_u[n]\big)\\
=&g\big(\tilde{\mathbf{P}}_u[n],N_E,\mathbf{Q}_E[n],t_{E,u}[n]\big)\\
&+
\text{tr}\big[ \mathcal{F}\big(\tilde{\mathbf{P}}_u[n],N_E,\mathbf{Q}_E[n],t_{E,u}[n]\big)\big(\mathbf{P}_u[n]-\tilde{\mathbf{P}}_u[n]\big)\big],
\end{aligned}
\end{equation}
where $g\big(\tilde{\mathbf{P}}_u[n],N_E,\mathbf{Q}_E[n],t_{E,u}[n]\big)=g\big(\mathbf{P}_u[n],N_E,$ $\mathbf{Q}_E[n],t_{E,u}[n]\big)|_{\mathbf{P}_u[n]=\tilde{\mathbf{P}}_u[n]}$.  Clearly, $g\big(\mathbf{P}_u[n],N_E,\mathbf{Q}_E[n],$ $t_{E,u}[n]|\tilde{\mathbf{P}}_u[n]\big)$ is a linear function in terms of $\mathbf{P}_u[n]$.

In this case, problem \eqref{subproblem3} in the $m$th iteration can be recast into
\begin{subequations}
		\small
	\label{subproblem3_approx}
	\allowdisplaybreaks
	\begin{align}
	\max_{\mathbf{\Phi}_u,\mathbf{\Phi}_a}  &\ \ \frac{1}{T_U}\sum_{n=1}^{N}\tau_n^{m-1}\Big[g\Big(\mathbf{P}_u[n],N_B,\mathbf{Q}_B[n],t^m_{B,u}[n]\Big)\\
	&\ \ \ \ \ \ \ \ \ \ \  \ \ \ \ \ \ \ \
	-g\Big(\mathbf{P}_a[n],N_B,\mathbf{Q}_B[n],t^m_{B,a}[n]\big|\mathbf{P}^{m-1}_a[n]\Big)\nonumber\\
	&\ \ \ \ \ \ \ \ \ \ \ \ \ \ \ \ \ \ \ -g\Big(\mathbf{P}_u[n],N_E,\mathbf{Q}_E[n],t^m_{E,u}[n]\big|\mathbf{P}^{m-1}_u[n]\Big) \nonumber\\
	&\ \ \ \ \ \ \ \ \ \ \ \ \ \ \ \ \ \ \
	+g\Big(\mathbf{P}_a[n],N_E,\mathbf{Q}_E[n],t^{m}_{E,a}[n]\Big)\Big]\nonumber\\	
	\text{s.t.} &  \sum_{n=1}^{N}p_{l}^u[n]\tau_n^{m-1}\leq E^{\max}, \forall l\\
	&\eqref{Pu_constraint},\eqref{Pu_Pa_constraint}, \eqref{Pa_constraint}.\ \ 
	\end{align}
\end{subequations}

Problem \eqref{subproblem3_approx} is convex in terms of $\mathbf{\Phi}_u$ and $\mathbf{\Phi}_a$. Therefore, the optimal variables $(\mathbf{\Phi}_u,\mathbf{\Phi}_a)$ can be solved by utilizing the standard optimization toolbox CVX.

(4) The transmission duration $\mathbf{\Xi}$

In the last step of the $m$th iteration, we focus on the transmission duration $\mathbf{\Xi}$ with the obtained $(\mathbf{t}_{B,u}^m,\mathbf{t}_{E,a}^m)$, $(\mathbf{t}_{B,a}^m,\mathbf{t}_{E,u}^m)$ and $(\mathbf{\Phi}_{u}^m,\mathbf{\Phi}_{a}^m)$, which can be achieved by the following optimization problem
\begin{subequations}
	\allowdisplaybreaks
	\label{subproblem4}
  		  	\small
	\begin{align}
	\max_{\mathbf{\Xi}}  & \frac{1}{T_U}\sum_{n=1}^{N}\tau_n\Big[g\Big(\mathbf{P}_u^m[n],N_B,\mathbf{Q}_B[n],t^m_{B,u}[n]\Big)\\
	&\ \ \ \ \ \ \ \ \ \ \ \ \  \ \ \ \ 
	-g\Big(\mathbf{P}_a^m[n],N_B,\mathbf{Q}_B[n],t^m_{B,a}[n]\Big)\nonumber\\
	&\ \ \ \ \ \ \ \ \ \ \ \ \ \ \ \ \  -g\Big(\mathbf{P}_u^m[n],N_E,\mathbf{Q}_E[n],t^m_{E,u}[n]\Big) \nonumber\\
	&\ \ \ \ \ \ \ \ \ \ \ \ \  \ \ \ \ 
	+g\Big(\mathbf{P}_a^m[n],N_E,\mathbf{Q}_E[n],t^{m}_{E,a}[n]\Big)\Big]\nonumber\\	
	\text{s.t.} & \ \ 
	\sum_{n=1}^{N}p_{l}^{u,m}[n]\tau_n\leq E^{\max}, \forall l \\
	& \sum_{n=1}^{N}\tau_n \leq T_t, \\
	&  0\leq \tau_n\leq \tau^{\max}, \forall n 
	\end{align}
\end{subequations}
which is a convex optimization, and can be solved by adopting the general toolbox.

Based on the above analysis, an overall iterative algorithm for problem \eqref{original_problem} can be achieved. Specifically, in each iteration, the original problem \eqref{original_problem} can be optimized by alternately solving problem \eqref{subproblem1}, problem \eqref{subproblem2}, problem \eqref{subproblem3} and problem \eqref{subproblem4} in an iterative manner. The details of the proposed algorithm can be summarized in Algorithm \ref{SCA_algorithm}. 

\begin{figure*}[ht]
	
	\normalsize
	\begin{equation}
	\label{eqn_taylor}
	\small
	\allowdisplaybreaks
	\begin{aligned}
	\allowdisplaybreaks
	&\bar{\mathcal{G}}\big[\mathbf{\Xi},\mathbf{\Phi}_{u},\mathbf{\Phi}_{a},\mathbf{t}_{B,u},\mathbf{t}_{B,a},\mathbf{t}_{E,u},\mathbf{t}
	_{E,a}\big|\big(\tilde{\mathbf{\Phi}}_u,\tilde{\mathbf{\Phi}}_a\big)\big]\\
	=&\frac{1}{T_U}\sum_{n=1}^{N}\tau_n\Big[g\big(\mathbf{P}_u[n],N_B,\mathbf{Q}_B[n],t_{B,u}[n]\big)
	-g\big(\mathbf{P}_a[n],N_B,\mathbf{Q}_B[n],t_{B,a}[n]|\tilde{\mathbf{P}}_a[n]\big)\\
	&\ \ \ \ \ \ \ \ \ \ \ \ \ \ \ -g\big(\mathbf{P}_u[n],N_E,\mathbf{Q}_E[n],t_{E,u}[n]|\tilde{\mathbf{P}}_u[n]\big)
	+g\big(\mathbf{P}_a[n],N_E,\mathbf{Q}_E[n],t_{E,a}[n]\big)\Big]\\
	=&\frac{1}{T_U}\sum_{n=1}^{N}\tau_n\Bigg\{\Big[\log_2 \det\Big(\mathbf{I}_L+\frac{N_B\mathbf{Q}_B[n]\mathbf{P}_u[n]\mathbf{Q}_B[n]}{\delta^2e^{t_{B,u}[n]}}\Big) +  N_B\log_2e\Big(t_{B,u}[n]-1+e^{-t_{B,u}[n]}\Big)\Big]\\
	& \ \ \ \ \ \ \ \ \ \ \ \ \ \ \ \ - \Big[\log_2 \det\Big(\mathbf{I}_L+\frac{N_B\mathbf{Q}_B[n]\tilde{\mathbf{P}}_a[n]\mathbf{Q}_B[n]}{\delta^2e^{t_{B,a}[n]}}\Big)+  N_B\log_2e\Big(t_{B,a}[n]-1+e^{-t_{B,a}[n]}\Big)\\
	& \ \ \ \ \ \ \ \ \ \ \ \ \ \ \ \ \ \ \ \ \ \ \ \ +\frac{\log_2eN_B}{\delta^2e^{t_{B,a}[n]}}\text{tr}\Big(\big[\mathbf{Q}_B[n]\big(\mathbf{I}_L+\frac{N_B}{\delta^2e^{t_{B,a}[n]}}\mathbf{Q}_B[n]\tilde{\mathbf{P}}_a[n]\mathbf{Q}_B[n]\big)^{-1}\mathbf{Q}_B[n]\big]\big(\mathbf{P}_a[n]-\tilde{\mathbf{P}}_a[n]\big)\Big) \Big]\\
	& \ \ \ \ \ \ \ \ \ \ \ \ \ \ \ \ -\Big[\log_2 \det\Big(\mathbf{I}_L+\frac{N_E\mathbf{Q}_E[n]\tilde{\mathbf{P}}_u[n]\mathbf{Q}_E[n]}{\delta^2e^{t_{E,u}[n]}}\Big) +  N_E\log_2e\Big(t_{E,u}[n]-1+e^{-t_{E,u}[n]}\Big)\\
	&\ \ \ \ \ \ \ \ \ \ \ \ \ \ \ \ \ \ \ \ \ \ \ \  + \frac{\log_2eN_E}{\delta^2e^{t_{E,u}[n]}}\text{tr}\Big(\big[\mathbf{Q}_E[n]\big(\mathbf{I}_L+\frac{N_E}{\delta^2e^{t_{E,u}[n]}}\mathbf{Q}_E[n]\tilde{\mathbf{P}}_u[n]\mathbf{Q}_E[n]\big)^{-1}\mathbf{Q}_E[n]\big]\big(\mathbf{P}_u[n]-\tilde{\mathbf{P}}_u[n]\big)\Big)\Big]\\
	& \ \ \ \ \ \ \ \ \ \ \ \ \ \ \ \ + \Big[\log_2 \det\Big(\mathbf{I}_L+\frac{N_E\mathbf{Q}_E[n]\mathbf{P}_a[n]\mathbf{Q}_E[n]}{\delta^2e^{t_{E,a}[n]}}\Big)
	+  N_E\log_2e\Big(t_{E,a}[n]-1+e^{-t_{E,a}[n]}\Big)\Big]\Bigg\}. 
	\end{aligned}
	\end{equation}
	\hrulefill
	\vspace*{4pt}
\end{figure*} 
\begin{algorithm}[t!]
	\caption{ The proposed iterative algorithm for solving problem $\eqref{original_problem}$  }
	\label{SCA_algorithm}
	\begin{algorithmic}[1]
		\STATE \textbf{Initialize}: the transmisson power $\mathbf{\Phi}_u^{0}=\big[\mathbf{P}_u^0[1],...,\mathbf{P}_u^0[N]\big]$, the AN power $\mathbf{\Phi}_a^{0}=\big[\mathbf{P}_a^0[1],...,\mathbf{P}_a^0[N]\big]$, the transmission duration $\mathbf{\Xi}^0=[\tau_1^0,\tau_2^0,...,\tau_N^0]$ and the accuracy $\epsilon>0$. Set $m=0$. \\ 
		\REPEAT
		\STATE Obtain the optimal set $(\mathbf{t}_{B,u}, \mathbf{t}_{E,a})$ by solving problem $\eqref{subproblem1}$ with the obtained sets ($\mathbf{\Phi}_u^{m-1}$, $\mathbf{\Phi}^{m-1}_{a}$) and $\mathbf{\Xi}^{m-1}$,\\
		\STATE Obtain the optimal set $(\mathbf{t}_{B,a}, \mathbf{t}_{E,u})$ by solving problem $\eqref{subproblem2}$ with the obtained sets ($\mathbf{\Phi}_u^{m-1}$, $\mathbf{\Phi}^{m-1}_{a}$) and $\mathbf{\Xi}^{m-1}$,\\
		\STATE Obtain the optimal set ($\mathbf{\Phi}_u$, $\mathbf{\Phi}_{a}$) by solving problem $\eqref{subproblem3}$ with the obtained sets $(\mathbf{t}^m_{B,u},\mathbf{t}^m_{E,a})$, $(\mathbf{t}^m_{B,a}, \mathbf{t}^m_{E,u})$ and $\mathbf{\Xi}^{m-1}$,\\
		\STATE Obtain the optimal set $\mathbf{\Xi}$ by solving problem $\eqref{subproblem4}$ with the obtained sets $(\mathbf{t}^m_{B,u}, \mathbf{t}^m_{E,a})$, $(\mathbf{t}^m_{B,a}, \mathbf{t}^m_{E,u})$ and ($\mathbf{\Phi}_u^{m}$, $\mathbf{\Phi}^{m}_{a}$),
		\STATE $m-1\leftarrow m$. \\
		\UNTIL The fractional increase of the objective function is below the threshold $\epsilon>0$.
	\end{algorithmic}
\end{algorithm}

\subsection{Convergence Performance Analysis}  
To analyze the convergence of the proposed algorithm, we
present the following Proposition 1.

\begin{prop}
	The approximation of the secrecy throughput is monotonically increasing in each iteration, i.e.,
	\begin{equation}
	R^{as}\big(\mathbf{\Xi}^{m},\mathbf{\Phi}_{u}^m,\mathbf{\Phi}_{a}^m\big)\geq 	R^{as}\big(\mathbf{\Xi}^{m-1},\mathbf{\Phi}_{u}^{m-1},\mathbf{\Phi}_{a}^{m-1}\big),
	\end{equation} 
	which demonstrates the convergence of the proposed algorithm.
\end{prop}

\begin{proof}
	Based on the analysis in Section III-B, we could achieve an approximation of $\mathcal{G}(\mathbf{\Xi},\mathbf{\Phi}_{u},\mathbf{\Phi}_{a},\mathbf{t}_{B,u},\mathbf{t}_{B,a},\mathbf{t}_{E,u},$ $\mathbf{t}
		_{E,a})$ as shown in \eqref{eqn_taylor}. Recall that any concave function is upper bounded by its first-order Taylor expansion at a given local point \cite{Boyd-2004}. The following upper-bounded expressions hold
	\begin{equation}
	\allowdisplaybreaks
	\small
	\label{Taylor_Ba}
	\begin{aligned}
	& g\big(\mathbf{P}_a[n],N_B,\mathbf{Q}_B[n],t_{B,a}[n]\big)\\
	\leq & g\big(\mathbf{P}_a[n],N_B,\mathbf{Q}_B[n],t_{B,a}[n]\big|\tilde{\mathbf{P}}_a[n]\big),
	\end{aligned}
	\end{equation} 
	and 
	\begin{equation}
	\small
	\label{Taylor_Eu}
	\begin{aligned}
	& g\big(\mathbf{P}_u[n],N_E,\mathbf{Q}_E[n],t_{E,u}[n]\big)\\
	\leq &  g\big(\mathbf{P}_u[n],N_E,\mathbf{Q}_E[n],t_{E,u}[n]\big|\tilde{\mathbf{P}}_u[n]\big),
	\end{aligned}
	\end{equation}
	where the equalities in \eqref{Taylor_Ba} and \eqref{Taylor_Eu} are met when $\mathbf{P}_a[n]=\tilde{\mathbf{P}}_a[n]$ and $\mathbf{P}_u[n]=\tilde{\mathbf{P}}_u[n]$, respectively.
	
	Thus, we could achieve 
	\begin{equation}
	\label{Taylor_for_objective}
	\begin{aligned}
	\allowdisplaybreaks
	& \mathcal{G}(\mathbf{\Xi},\mathbf{\Phi}_{u},\mathbf{\Phi}_{a},\mathbf{t}_{B,u},\mathbf{t}_{B,a},\mathbf{t}_{E,u},\mathbf{t}	
	_{E,a})\\
	\geq &\bar{\mathcal{G}}\big[\mathbf{\Xi},\mathbf{\Phi}_{u},\mathbf{\Phi}_{a},\mathbf{t}_{B,u},\mathbf{t}_{B,a},\mathbf{t}_{E,u},\mathbf{t}
	_{E,a}\big|\big(\tilde{\mathbf{\Phi}}_u,\tilde{\mathbf{\Phi}}_a\big)\big],
	\end{aligned}
	\end{equation}
	where the equality holds when $\mathbf{\Phi}_{u}=\tilde{\mathbf{\Phi}}_u$ and $\mathbf{\Phi}_{a}=\tilde{\mathbf{\Phi}}_a$.
	
	Based on the fact \eqref{Taylor_for_objective}, we present the convergence of Algorithm 1, as shown next.      
	In the $m$th iteration, the optimal solutions ($\mathbf{\Phi}_u^{m}$, $\mathbf{\Phi}^{m}_{a}$), ($\mathbf{t}
	_{B,u}^{m}$, $\mathbf{t}_{E,a}^{m}$), ($\mathbf{t}
	_{B,a}^{m}$, $\mathbf{t}_{E,u}^{m}$) and $\mathbf{\Xi}^m$ can be obtained by Algorithm 1.
	Based on the properties of the saddle point \cite{Brayton-1979}, the following relationship in the $m$th iteration holds
	\begin{equation}
	\small
	\allowdisplaybreaks
	\label{aaa}
	\begin{aligned}
	&\bar{\mathcal{G}}\big[\mathbf{\Xi}^m,\mathbf{\Phi}_{u}^m,\mathbf{\Phi}_{a}^m,\mathbf{t}_{B,u}^m,\mathbf{t}_{B,a}^m,\mathbf{t}_{E,u}^m,\mathbf{t}
	_{E,a}^m\big|\big(\tilde{\mathbf{\Phi}}_u^{m-1},\tilde{\mathbf{\Phi}}_a^{m-1}\big)\big]\\
	\geq &  \bar{\mathcal{G}}\big[\mathbf{\Xi},\mathbf{\Phi}_{u},\mathbf{\Phi}_{a},\mathbf{t}_{B,u}^m,\mathbf{t}_{B,a}^m,\mathbf{t}_{E,u}^m,\mathbf{t}
	_{E,a}^m\big|\big(\tilde{\mathbf{\Phi}}_u^{m-1},\tilde{\mathbf{\Phi}}_a^{m-1}\big)\big].
	\end{aligned}
	\end{equation}
	
	Let $\mathbf{\Phi}_{u}=\mathbf{\Phi}_{u}^{m-1}$, $\mathbf{\Phi}_{a}=\mathbf{\Phi}_{a}^{m-1}$ and $\mathbf{\Xi}=\mathbf{\Xi}^{m-1}$. Then, it follows from \eqref{aaa} that 
	\begin{equation}
	\small
	\allowdisplaybreaks
	\label{abab}
	\begin{aligned}
	&\bar{\mathcal{G}}\big[\mathbf{\Xi}^m,\mathbf{\Phi}_{u}^m,\mathbf{\Phi}_{a}^m,\mathbf{t}_{B,u}^m,\mathbf{t}_{B,a}^m,\mathbf{t}_{E,u}^m,\mathbf{t}
	_{E,a}^m\big|\big(\tilde{\mathbf{\Phi}}_u^{m-1},\tilde{\mathbf{\Phi}}_a^{m-1}\big)\big]\\
	\geq & \bar{\mathcal{G}}\big[\mathbf{\Xi}^{m-1},\mathbf{\Phi}_{u}^{m-1},\mathbf{\Phi}_{a}^{m-1},\mathbf{t}_{B,u}^m,\mathbf{t}_{B,a}^m,\mathbf{t}_{E,u}^m,\mathbf{t}
	_{E,a}^m\big|\big(\tilde{\mathbf{\Phi}}_u^{m-1},\tilde{\mathbf{\Phi}}_a^{m-1}\big)\big]\\
	\stackrel{(a)}{=}& \mathcal{G}(\mathbf{\Xi}^{m-1},{\Phi}_{u}^{m-1},\mathbf{\Phi}_{a}^{m-1},\mathbf{t}_{B,u}^{m},\mathbf{t}_{B,a}^{m},\mathbf{t}_{E,u}^{m},\mathbf{t}	
	_{E,a}^{m}),
	\end{aligned}
	\end{equation} 
	where (a) holds according to the analysis in \eqref{Taylor_for_objective}.  
	
	According to \eqref{R_ap}, we know that \footnote{Due to the limited space, we simplify the constraint $\mathbf{t}_{q,x}\geq \mathbf{0}$ as $\mathbf{t}_{q,x}$, where $q\in \{B,E\}$ and $x \in \{u,a\}$.}
	\begin{equation}
	\allowdisplaybreaks
	\small
	\label{R_ap_convergence}
	\begin{aligned}
	&R^{as}\big(\mathbf{\Xi}^{m},\mathbf{\Phi}_{u}^m,\mathbf{\Phi}_{a}^m\big)\\ 
	=&
	\begin{aligned}
	\max_{\tiny
		\mathbf{t}_{B,a},
		\mathbf{t}_{E,u}}
	\end{aligned} 
	\begin{aligned} \
	\min_{\tiny
		\mathbf{t}_{B,u},
		\mathbf{t}_{E,a}}
	\end{aligned} 
	\ \   \mathcal{G}(\mathbf{\Xi}^{m},\mathbf{\Phi}_{u}^m,\mathbf{\Phi}_{a}^m,\mathbf{t}_{B,u},\mathbf{t}_{B,a},
	\mathbf{t}_{E,u},\mathbf{t}
	_{E,a})\\
	\stackrel{(b)}{=}& \begin{aligned} \
	\min_{\tiny
		\mathbf{t}_{B,u},
		\mathbf{t}_{E,a}}
	\end{aligned} 
	\mathcal{G}(\mathbf{\Xi}^{m},\mathbf{\Phi}_{u}^m,\mathbf{\Phi}_{a}^m,\mathbf{t}_{B,u},\mathbf{t}_{B,a}^m,
	\mathbf{t}_{E,u}^m,\mathbf{t}_{E,a})\\
	\stackrel{(c)}{\geq} &
	\begin{aligned}
	\min_{\tiny
		\mathbf{t}_{B,u},
		\mathbf{t}_{E,a}}
	\end{aligned} 
	\bar{\mathcal{G}}\big[\mathbf{\Xi}^{m},\mathbf{\Phi}_{u}^m,\mathbf{\Phi}_{a}^m,\mathbf{t}_{B,u},\mathbf{t}_{B,a}^m,\mathbf{t}_{E,u}^m,\mathbf{t}
	_{E,a}\big|\big(\tilde{\mathbf{\Phi}}_u^{m-1},\tilde{\mathbf{\Phi}}_a^{m-1}\big)\big],
	\end{aligned}   
	\end{equation}    
	where step (b) holds since ($\mathbf{t}
	_{B,a}^{m}$, $\mathbf{t}_{E,u}^{m}$) are the optimal solutions by using Algorithm 1, and step (c) holds due to the first-order Taylor expansion as shown in \eqref{Taylor_for_objective}.
	
	Furthermore, we have
	\begin{equation}
	\allowdisplaybreaks
	\label{xxx}
	\small
	\begin{aligned}
	&\bar{\mathcal{G}}\big[\mathbf{\Xi}^{m},\mathbf{\Phi}_{u}^m,\mathbf{\Phi}_{a}^m,\mathbf{t}_{B,u}^m,\mathbf{t}_{B,a}^m,\mathbf{t}_{E,u}^m,\mathbf{t}
	_{E,a}^m\big|\big(\tilde{\mathbf{\Phi}}_u^{m-1},\tilde{\mathbf{\Phi}}_a^{m-1}\big)\big]\\
	=&
	\begin{aligned}
	\min_{\tiny
		\mathbf{t}_{B,u},
		\mathbf{t}_{E,a}}
	\end{aligned} 
	\bar{\mathcal{G}}\big[\mathbf{\Xi}^{m},\mathbf{\Phi}_{u}^m,\mathbf{\Phi}_{a}^m,\mathbf{t}_{B,u},\mathbf{t}_{B,a}^m,\mathbf{t}_{E,u}^m,\mathbf{t}
	_{E,a}\big|\big(\tilde{\mathbf{\Phi}}_u^{m-1},\tilde{\mathbf{\Phi}}_a^{m-1}\big)\big].
	\end{aligned}
	\end{equation}
	
	Thus, according to \eqref{R_ap}, \eqref{abab}, \eqref{R_ap_convergence} and \eqref{xxx}, it follows that
	\begin{equation}
	\label{convergence_proof}
	\allowdisplaybreaks
	\small 
	\begin{aligned}
	\allowdisplaybreaks
	&R^{as}\big(\mathbf{\Xi}^{m},\mathbf{\Phi}_{u}^m,\mathbf{\Phi}_{a}^m\big)\\
	\geq& \bar{\mathcal{G}}\big[\mathbf{\Xi}^{m},\mathbf{\Phi}_{u}^m,\mathbf{\Phi}_{a}^m,\mathbf{t}_{B,u}^m,\mathbf{t}_{B,a}^m,\mathbf{t}_{E,u}^m,\mathbf{t}
	_{E,a}^m\big|\big(\tilde{\mathbf{\Phi}}_u^{m-1},\tilde{\mathbf{\Phi}}_a^{m-1}\big)\big] \\
	\geq &
	\mathcal{G}(\mathbf{\mathbf{\Xi}^{m-1},\Phi}_{u}^{m-1},\mathbf{\Phi}_{a}^{m-1},\mathbf{t}_{B,u}^{m},\mathbf{t}_{B,a}^{m},\mathbf{t}_{E,u}^{m},\mathbf{t}	
	_{E,a}^{m})\\
	\stackrel{(d)}{=}
	&\begin{aligned}
	\max_{\tiny
		\mathbf{t}_{B,a},
		\mathbf{t}_{E,u}
	}
	\end{aligned} 
	\begin{aligned} \
	\min_{\tiny
		\mathbf{t}_{B,u},
		\mathbf{t}_{E,a}
	}
	\end{aligned} 
	\ \  
	\mathcal{G}(\mathbf{\Xi}^{m-1},\mathbf{\Phi}_{u}^{m-1},\mathbf{\Phi}_{a}^{m-1},\mathbf{t}_{B,u},\mathbf{t}_{B,a},
	\mathbf{t}_{E,u},\mathbf{t}
	_{E,a})\\
	\stackrel{(e)}{=}
	&R^{as}\big(\mathbf{\Xi}^{m-1},\mathbf{\Phi}_{u}^{m-1},\mathbf{\Phi}_{a}^{m-1}\big),
	\end{aligned}
	\end{equation}
	where step (d) holds since ($\mathbf{t}_{B,u}$, 
	$\mathbf{t}_{E,a}$) and ($\mathbf{t}_{B,a}$, 
	$\mathbf{t}_{E,u}$) are the optimal solutions in Algorithm 1, and step (e) holds according to the closed form of the secrecy throughput in \eqref{R_ap}.
	
	\eqref{convergence_proof} indicates that $R^{as}\big(\mathbf{\Xi},\mathbf{\Phi}_{u},\mathbf{\Phi}_{a}\big)$ is nondecreasing in each iteration, which can assure the convergence of Algorithm 1. This completes the proof.
\end{proof}
\vspace{-0.2cm}
\section{Numerical Results and Discussions}
In this section, the performance of our proposed scheme is verified by simulation. We consider a $1000$m $\times$ $1000$m square cell, and suppose there are $N=10$ legitimate receivers which are randomly distributed in the cell. The multiple single-antenna rotary-wing UAVs, forming a UAV swarm, hover above the scheduled legitimate receiver to provide the confidential messages during a transmission duration. When the UAVs hover to serve the scheduled one, they are randomly dispatched in the circle of radius $50$m with altitude $100$m $\sim$ $200$m, which is modeled as a cylinder in Fig. \ref{AveSecRate_cylinder}. Then, the UAVs fly to the next scheduled one within a flying duration. The eavesdropper randomly locates with a safety distance $r_e=100$m away from the scheduled legitimate user, which is modeled as a circle in Fig. \ref{AveSecRate_cylinder}, to hide himself, and moves following the trajectory of the UAV swarm for better eavesdropping. Unless otherwise specified, the system parameters are set as follows: the number of antennas for the legitimate users $N_B=5$, the number of antennas for the eavesdropper $N_E=3$, the transmission duration budget $\tau^{\max}=8s$, the total transmission duration $T_t=100s$, the consecutive period $T_U=210s$ \cite{Wu-TWC2018}, $f=2.4$GHz \cite{Xuan-JCIN2018}, $c=3\times 10^8$m/s, $a=5.0188$, $b=0.3511$  \cite{Al-Hourani-2014WCL} and the noise covariance $\delta^2=-107$dBm \cite{Feng-2013}. The threshold presented in Algorithm 1 is fixed as $\epsilon=10^{-3}$. We consider the typical propagation environments using the following $(\eta_{\text{LoS}},\eta_{\text{NLoS}})$ pairs $(0.1, 21)$, $(1.0, 20)$, $(1.6, 23)$, $(2.3, 34)$ corresponding to suburban, urban, dense urban, and highrise urban, respectively \cite{Al-Hourani-2014WCL}.


\begin{figure} [t!]  
	\centering 		 		 
	\includegraphics[width=0.8\columnwidth]{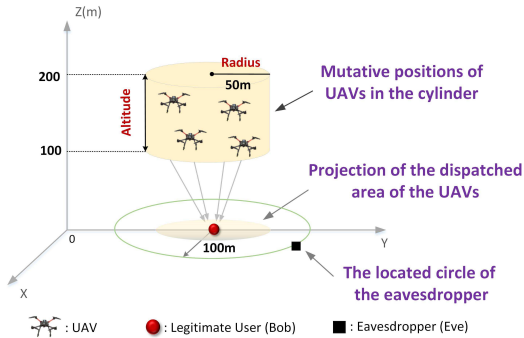}       	 
	\caption{Illustration of positions of UAVs and the eavesdropper during a transmission duration. The UAVs hover at arbitrary positions in the cylinder above the scheduled legitimate user to serve it. The eavesdropper locates in a circle with a safety distance $100$m away from the legitimate user.} 
	\label{AveSecRate_cylinder}
\end{figure}


\begin{figure}
	\centering 		 		 
	\includegraphics[width=0.8\columnwidth]{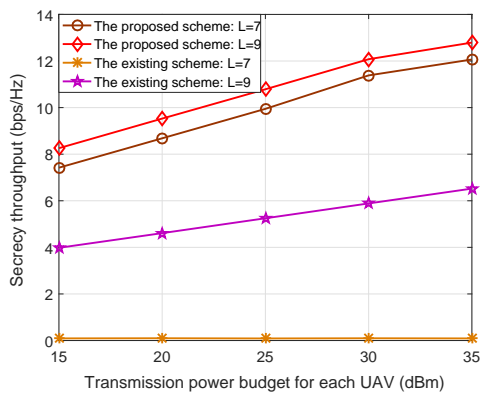}       	 
	\caption{Secrecy throughput based on different transmission strategies.} 
	\label{AveSecRate_comparsion}
\end{figure}

To depict the performance of the proposed scheme, we compare it with the existing scheme in Fig. \ref{AveSecRate_comparsion} by considering the suburban environment. We assume
$\mathbf{P}_u^0[n]=\bar{p}_u\mathbf{I}_L, \forall n$, where $\bar{p}_u=20$dBm, $\mathbf{P}_a^0[n]=\bar{p}_a\mathbf{I}_L, \forall n$, where $\bar{p}_a=10$dBm, $\mathbf{\Xi}^0=[1,1,...,1]$, $E^{\max}=150$J. In the existing scheme, confidential messages are transmitted for the legitimate receivers by UAVs, and AN is transmitted in the null space of the legitimate channel according to the instantaneous CSI $\mathbf{H}_B$ via precoding. Similar to \cite{2016TVT-Na}, the power allocation between the desired signals and AN is considered, where the ratio $\phi=\frac{N_B}{N_E+N_B}$ of the transmission power budget is allocated to the desired signals, and the ratio $1-\phi$ is allocated to AN. Note that, this scheme has been widely investigated in the existing literatures but the large-scale CSI has not been taken into account. From the simulation results, we can observe that the proposed scheme presents a significant performance gain over the existing scheme in the case $L=7$ and 9, which can be explained as follows. In our proposed scheme, the UAV swarm could adaptively transmit the confidential messages in a higher power when it is close to the legitimate users based on the large-scale CSI, and allocate the higher power for AN when the swarm is close to the eavesdropper, which promotes the secrecy performance improvement. However, due to the inflexible signal transmission mode in the existing scheme, the confidential messages and AN are transmitted in orthogonal channel spaces. Although the eavesdropper can be well suppressed by exploiting AN, it is not able to improve receiving quality of the legitimate receivers with the fixed power of the desired signals. Therefore, a poor secrecy performance is achieved.

To further illustrate the convergence of Algorithm 1, we present the convergence process for $100$ randomly-generated system topologies by the proposed Algorithm 1 in the suburban environment in Fig. \ref{Iteration}. We initialize 
$\mathbf{P}_u^0[n]=\bar{p}_u\mathbf{I}_L, \forall n$, where $\bar{p}_u=30$dBm, $\mathbf{P}_a^0[n]=\bar{p}_a\mathbf{I}_L, \forall n$, where $\bar{p}_a=0$dBm, $\mathbf{\Xi}^0=[1,1,...,1]$, and set   
the transmission power budget $P^{\max}=30$dBm, the transmission energy budget $E^{\max}=300$J, the number of the UAVs $L=7$. It can be observed that for the most cases, Algorithm 1 can converge within $6$ iterations and the ratio of iterations between $4 \sim 5$ could achieve $86\% $,  which demonstrates the validity of the proposed scheme.


\begin{figure}  
	\centering       		 			 
	\includegraphics[width=0.8\columnwidth]{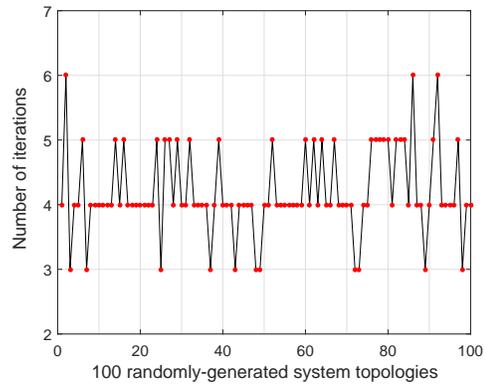}        		  		
	\caption{Convergence of the proposed algorithm.} 
	\label{Iteration}
	\vspace{-0.5cm}
\end{figure}

Fig. \ref{AverSecRate_vs_diffPS} illustrates secrecy throughput by Algorithm 1 versus the transmission power budget for each UAV in the suburban scenario. We assume $\mathbf{P}_u^0[n]=\bar{p}_u\mathbf{I}_L, \forall n$, where $\bar{p}_u=20$dBm, $\mathbf{P}_a^0[n]=\bar{p}_a\mathbf{I}_L, \forall n$, where $\bar{p}_a=5$dBm, $\mathbf{\Xi}^0=[1,1,...,1]$, $E^{\max}=150$J. From Fig. \ref{AverSecRate_vs_diffPS}, we can see that secrecy throughput increases when the transmission power budget for each UAV becomes large. That is due to the fact that the increasing transmission power can enhance the achievable ergodic rate at the legitimate user or the eavesdropper. Furthermore, in our proposed scheme, the confidential message and AN are transmitted independently without cooperation at each UAV. Based on the large-scale CSI, each UAV could allocate more power to the confidential message to improve the achievable rate of Bob when the UAV swarm is close to the legitimate users. When the swarm is close to Eve, more power would be allocated to AN for decreasing the signal receiving quality of Eve. Also, for the same transmission power budget of each UAV, secrecy throughput increases as the number of UAVs in the swarm. Obviously, the larger the number of UAVs is, the higher the power of the confidential messages is, which significantly enhances the secrecy performance of the system.   


\begin{figure}[t!]
	\centering       		 			 
	\includegraphics[width= 0.8\columnwidth]{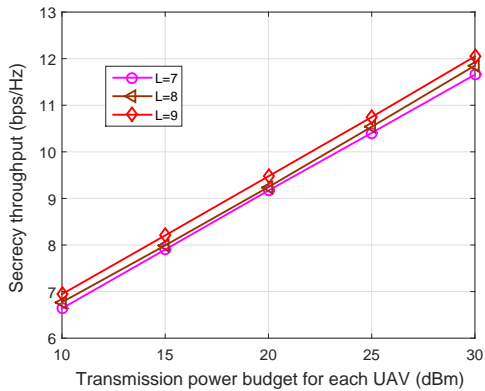}       	 
	\caption{Secrecy throughput versus transmission power budget for each UAV.} 
	\label{AverSecRate_vs_diffPS}
\end{figure} 

Fig. \ref{AverSecRate_vs_diffES} depicts secrecy throughput achieved by Algorithm 1 versus the total transmission energy budget for each UAV in the suburban scenario. It is assumed that $\mathbf{P}_u^0[n]=\bar{p}_u\mathbf{I}_L, \forall n$, where $\bar{p}_u=23$dBm, $\mathbf{P}_a^0[n]=\bar{p}_a\mathbf{I}_L, \forall n$, where $\bar{p}_a=10$dBm, $\mathbf{\Xi}^0=[1,1,...,1]$, and $P^{\max}=35$dBm. From Fig. \ref{AverSecRate_vs_diffES}, it can be seen that secrecy throughput increases as the total transmission energy constraint for each UAV becomes large. That is because when the transmission energy budget at each UAV increases, the transmission power at each UAV is getting large over its flight. By using the large-scale CSI, the power of the confidential messages and the AN power can be intelligently designed. More power can be transmitted for the confidential messages when the swarm is close to Bob, and more power is allocated for AN when the swarm is near to Eve. Thus, a positive secrecy throughput can be achieved. Furthermore, we can also see that at the same transmission energy budget of each UAV, secrecy throughput grows with the increasing number of UAVs in the swarm. That is due to the fact that as the number of UAVs increases, the total transmission energy of the UAV swarm is getting high. In this case, secrecy throughput could increase. 

\begin{figure}[t!]
	\centering       		 			 
	\includegraphics[width= 0.8\columnwidth]{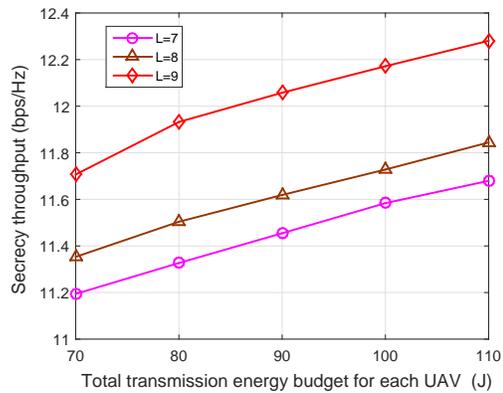}       	 
	\caption{Secrecy throughput versus total transmission energy budget for each UAV.} 
	\label{AverSecRate_vs_diffES}
\end{figure}


\begin{figure}[t!]
	\centering       		 			 
	\includegraphics[width= 0.8\columnwidth]{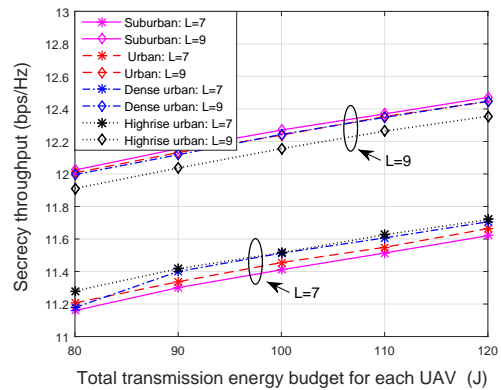}       	 
	\caption{ Secrecy throughput versus total transmission energy budget for each UAV in different typical environments.} 
	\label{AveSecRate_diffPs_Environments}
\end{figure}


Fig. \ref{AveSecRate_diffPs_Environments} presents secrecy throughput obtained by Algorithm 1 versus the total transmission energy budget for each UAV in different urban scenarios. We specify the initial transmission power  $\mathbf{P}_u^0[n]=\bar{p}_u\mathbf{I}_L, \forall n$, where $\bar{p}_u=25$dBm, $\mathbf{P}_a^0[n]=\bar{p}_a\mathbf{I}_L, \forall n$, where $\bar{p}_a=10$dBm, $P^{\max}=35$dBm and $\mathbf{\Xi}^0=[1,1,...,1]$. It can be observed that secrecy throughput is various in the different propagation environments. That comes from that the power loss is highly dependent on the practical urban environments \cite{Al-Hourani-2014WCL}, which leads to the divergence of secrecy throughput. Also, it can be seen that as the total transmission energy budget for each UAV becomes large, secrecy throughput increases. Furthermore, secrecy throughput in the case $L=9$ is larger than that in the case $L=7$. That is because that the total transmission energy of the UAV swarm increases with the number of UAVs, which leads to the higher secrecy throughput.

\section{Conclusions}
In this paper, we investigated power allocation of AN-assist secure transmission for the UAV swarm-enabled aerial CoMP, where the eavesdropper moved following the trajectory of the swarm for better eavesdropping. We took the composite channel including small-scale and large-scale fading into account. Considering the hardness in acquiring the perfect CSI, we maximized secrecy throughput by utilizing only the large-scale CSI of the legitimate receivers and the eavesdropper. Specifically, we designed the problem by jointly optimizing the transmission power of the desired signals and the AN power under the constraints of the total transmission energy of each rotary-wing UAV during its flying period and the transmission durations for all the legitimate users. The formulated problem was a non-convex one. To handle that, we first achieved a closed form of the secrecy throughput, and then provided an iterative algorithm for the problem. Finally, we evaluated the effectiveness of our proposed iterative algorithm by means of simulation results. 

While an efficient power allocation scheme has been developed to enhance the secure transmission for the UAV swarm-enabled aerial CoMP, in future works we plan to consider a pre-processing framework for the UAV swarm-enabled aerial networks, where the user scheduling, UAV trajectory and UAV altitude will be jointly considered, in order to provide guidelines for the secure design of the system. Furthermore, we will also develop a novel power allocation method based on \cite{TSG2019-ZhenLi,TFS2019-ZhenLi,TSG2018-ZhenLi,TCS2018-ZhenLi} by simultaneously considering the blockage caused by the abundant transmission and the transmission energy constraint to improve the secrecy performance of the system.
\vspace{-0.5cm}


%

\appendices
\section{Proof of Theorem 1}
\label{appendix_theorem_Rap}
Considering the identical structure of four terms in \eqref{eqn}, we focus on the first term and denote it as $f_{B,u}(\mathbf{\Xi},\mathbf{\Phi}_u)=$ $\frac{1}{T_U} \sum_{n=1}^{N}\tau_n\mathbb{E}_{\mathbf{S}_{B}[n]}\Big[\log_2 \det \Big(\mathbf{I}_{N_B}+\frac{\mathbf{H}_{B}[n]\mathbf{P}_{u}[n](\mathbf{H}_{B}[n])^H}{\delta^2}\Big)\Big]$ for clarity. Referring to the remarkable studies in \cite{Feng-2013, Feng-WOCC}, the closed form of $f_{B,u}(\mathbf{\Xi},\mathbf{\Phi}_u)$ could be expressed as
\begin{equation} 
\small
\label{approx_eq}
\begin{aligned}
&f_{B,u}(\mathbf{\Xi},\mathbf{\Phi}_u)\\
\approx& \frac{1}{T_U}\sum_{n=1}^{N}\tau_n\Bigg[ \log_2 \det\Big(\mathbf{I}_L+\frac{N_B\mathbf{Q}_B[n]\mathbf{P}_u[n]\mathbf{Q}_B[n]}{\delta^2w_{B,u}[n]}\Big)\\
& \ \   +  N_B\Big[\log_2(w_{B,u}[n])-\log_2e\big(1-\frac{1}{w_{B,u}[n]}\big)\Big]\Bigg]\\
\triangleq& y\big(\mathbf{\Xi},\mathbf{\Phi}_u,\mathbf{w}_{B,u}\big),
\end{aligned}
\end{equation}
where $\mathbf{w}_{B,u}=\{w_{B,u}[n]\geq 1, \forall n\}$, and $w_{B,u}[n], \forall n$ can be uniquely determined by the following fixed-point equation
\begin{equation}
\label{fixed_point_equ}
\begin{aligned}
\small
w_{B,u}[n]
=1+\sum_{l=1}^{L}\frac{p_l^{u}[n]Q_{B,l}^{-1}[n]}{\delta^2+N_Bp_l^{u}[n]\big(Q_{B,l}[n]w_{B,u}[n]\big)^{-1}}, \forall n.
\end{aligned}
\end{equation} 

Note that $y\big(\mathbf{\Xi},\mathbf{\Phi}_u,\mathbf{w}_{B,u}\big)$ is a quite accurate form of $f_{B,u}(\mathbf{\Xi},\mathbf{\Phi}_u)$ \cite{Feng-2013}. However, $y\big(\mathbf{\Xi},\mathbf{\Phi}_u,\mathbf{w}_{B,u}\big)$ is still intractable since $w_{B,u}[n], \forall n$ involves the variable $p_l^u[n], \forall n$. By introducing $\mathbf{t}_{B,u}=\{t_{B,u}[n], \forall n\}$ and specifying $w_{B,u}[n]=e^{t_{B,u}[n]}, \forall n$, we have 
\begin{equation} 
\small
\label{approx_eq1}
\begin{aligned}
&\tilde{y}\big(\mathbf{\Xi},\mathbf{\Phi}_u,\mathbf{t}_{B,u}\big)\\
=&\frac{1}{T_U}\sum_{n=1}^{N}\tau_n\Bigg[ \log_2 \det\Big(\mathbf{I}_L+\frac{N_B\mathbf{Q}_B[n]\mathbf{P}_u[n]\mathbf{Q}_B[n]}{\delta^2e^{t_{B,u}[n]}}\Big)\\
& \ \ \ \ \ \ \ \ \ \  + N_B\log_2e\Big(t_{B,u}[n]-1+e^{-t_{B,u}[n]}\Big)\Bigg]\\
\triangleq&\frac{1}{T_U}\sum_{n=1}^{N}\tau_n g\Big(\mathbf{P}_u[n],N_B,\mathbf{Q}_B[n],t_{B,u}[n]\Big),
\end{aligned}
\end{equation}
where $\mathbf{t}_{B,u}\geq \mathbf{0}$ and $e^{t_{B,u}[n]}, \forall n$ needs to satisfy \eqref{fixed_point_equ}. 

Based on \eqref{approx_eq1}, we first consider the partial derivation of $\tilde{y}\big(\mathbf{\Xi},\mathbf{\Phi}_u,\mathbf{t}_{B,u}\big)$ with respect to $t_{B,u}[n]$, which can be expressed as
\begin{equation}
\allowdisplaybreaks
\label{first_derivation}
\small
\begin{aligned}
\allowdisplaybreaks
&\frac{\partial\tilde{y}\big(\mathbf{\mathbf{\Xi},\Phi}_u,\mathbf{t}_{B,u}\big)}{\partial t_{B,u}[n]}\\
=&-\frac{\tau_nN_B}{T_U\ln2}\Bigg[\sum_{l=1}^{L}\frac{p_l^u[n]\big(Q_{B,l}[n]e^{t_{B,u}[n]}\big)^{-1}}{\delta^2+N_Bp_l^u[n]\big(Q_{B,l}[n]e^{t_{B,u}[n]}\big)^{-1}}\\
& \ \ \ \ \ \ \ \ \ \ \ \ \ \ \ \ \ \ \ \ \ \ \ \ \ \ \ \ \  -1+e^{-t_{B,u}[n]}\Bigg], \forall n.
\end{aligned} 
\end{equation}

Let $s(t_{B,u}[n])=\sum_{l=1}^{L}\frac{p_l^u[n]\big(Q_{B,l}[n]e^{t_{B,u}[n]}\big)^{-1}}{\delta^2+N_Bp_l^u[n]\big(Q_{B,l}[n]e^{t_{B,u}[n]}\big)^{-1}}-1+e^{-t_{B,u}[n]}$. Thus, we have
\begin{equation}
\small
\label{monotonically_proof}
\begin{aligned}
\frac{\partial\tilde{y}\big(\mathbf{\Xi},\mathbf{\Phi}_u,\mathbf{t}_{B,u}\big)}{\partial t_{B,u}[n]}
\begin{cases}
<0, & s(t_{B,u}[n])>0,\\
=0, & s(t_{B,u}[n])=0,\\
>0, & s(t_{B,u}[n])<0. 
\end{cases}
\end{aligned}
\end{equation} 
Denote $\mathbf{s}(\mathbf{t}_{B,u})=\big[s(t_{B,u}[1]),\cdots,s(t_{B,u}[N])\big]^T$.
It can be observed from \eqref{monotonically_proof} that
$\tilde{y}\big(\mathbf{\Xi},\mathbf{\Phi}_u,\mathbf{t}_{B,u}\big)$ is monotonically decreasing when $\mathbf{s}(\mathbf{t}_{B,u})>\mathbf{0}$, $\tilde{y}\big(\mathbf{\Xi},\mathbf{\Phi}_u,\mathbf{t}_{B,u}\big)$ is monotonically increasing when $\mathbf{s}(\mathbf{t}_{B,u})<\mathbf{0}$, and  $\tilde{y}\big(\mathbf{\Xi},\mathbf{\Phi}_u,\mathbf{t}_{B,u}\big)$ achieves the extreme minimum point when $\mathbf{s}(\mathbf{t}_{B,u})=\mathbf{0}$. That is, when $\mathbf{t}_{B,u}$ increases from $\mathbf{0}$ to infinity, the function  $\tilde{y}\big(\mathbf{\mathbf{\Xi},\Phi}_u,\mathbf{t}_{B,u}\big)$ would first decrease and then increase.

In addition, it can be derived from \eqref{first_derivation} that 
\begin{equation}
\allowdisplaybreaks
\small
\begin{aligned}
\allowdisplaybreaks
&\frac{\partial^2\tilde{y}\big(\mathbf{\Xi},\mathbf{\Phi}_u,\mathbf{t}_{B,u}\big)}{\partial t_{B,u}^2[n]}\\
=&\frac{\tau_nN_B}{T_U\ln2}\Bigg[\sum_{l=1}^{L}\frac{\delta^2p_l^u[n]Q_{B,l}^{-1}[n]e^{t_{B,u}[n]}}{\big(\delta^2e^{t_{B,u}[n]}+N_Bp_l^u[n]Q_{B,l}^{-1}[n]\big)^2}+e^{-t_{B,u}[n]}\Bigg]\\
>&0.
\end{aligned} 
\end{equation}

Therefore, we can obtain that
\begin{equation}
\small
\allowdisplaybreaks
\label{second_order_derivation}
\begin{aligned}
\allowdisplaybreaks
&\frac{\partial^2\tilde{y}\big(\mathbf{\Xi},\mathbf{\Phi}_u,\mathbf{t}_{B,u}\big)}{\partial \mathbf{t}_{B,u}^2}\\
=&\Bigg[\frac{\partial^2\tilde{y}\big(\mathbf{\Xi},\mathbf{\Phi}_u,\mathbf{t}_{B,u}\big)}{\partial t_{B,u}^2[1]},\cdots,\frac{\partial^2\tilde{y}\big(\mathbf{\Xi},\mathbf{\Phi}_u,\mathbf{t}_{B,u}\big)}{\partial t_{B,u}^2[N]}\Bigg]^T
\geq \mathbf{0},
\end{aligned}
\end{equation} 
which indicates that  $\tilde{y}\big(\mathbf{\Xi},\mathbf{\Phi}_u,\mathbf{t}_{B,u}\big)$ is convex in terms of $\mathbf{t}_{B,u}$.

With the above discussions from \eqref{first_derivation} to \eqref{second_order_derivation}, we can conclude that the function  $\tilde{y}\big(\mathbf{\Xi},\mathbf{\Phi}_u,\mathbf{t}_{B,u}\big)$ would achieve the minimum value if and only if $\mathbf{t}_{B,u}$ satisfies the following equation
\begin{equation}
\label{fixed_point_equ1}
\sum_{l=1}^{L}\frac{p_l^u[n]\big(Q_{B,l}[n]e^{t_{B,u}[n]}\big)^{-1}}{\delta^2+N_Bp_l^u[n]\big(Q_{B,l}[n]e^{t_{B,u}[n]}\big)^{-1}}-1+e^{-t_{B,u}[n]}=0, \forall n.
\end{equation}

Comparing \eqref{fixed_point_equ} and \eqref{fixed_point_equ1}, it can be observed that when $\mathbf{t}_{B,u}$ satisfys the fixed-point equation, the function $y\big(\mathbf{\Xi},\mathbf{\Phi}_u,\mathbf{w}_{B,u}\big)$ would be minimized. Therefore, we have
\begin{equation}
\label{max_min_equ}
y\big(\mathbf{\mathbf{\Xi},\Phi}_u,\mathbf{w}_{B,u}\big)=\min_{\mathbf{t}_{B,u}\geq \mathbf{0}}\tilde{y}\big(\mathbf{\Xi},\mathbf{\Phi}_u,\mathbf{t}_{B,u}\big).
\end{equation}

Finally, we generalize the aforementioned derivation to the other three terms in \eqref{eqn}, and then complete the proof.
\vspace{-0.3cm}

\ifCLASSOPTIONcaptionsoff
  \newpage
\fi

\end{document}